\documentclass[12pt]{article}
\usepackage{amsfonts}
\usepackage{eurosym}
\usepackage{graphicx}
\usepackage{booktabs}
\usepackage{float}
\usepackage{subfig}
\usepackage{amssymb}
\usepackage{comment}
\usepackage{amsmath}
\usepackage{booktabs}
\usepackage{multirow}
\usepackage{enumerate}
\usepackage{color}
\usepackage{xcolor}
\usepackage[authoryear]{natbib}
\usepackage{hyperref}

\newtheorem{lemma}{{\bf \sc Lemma}}

\newtheorem{corollary}{{\bf \sc Corollary}}
\newtheorem{proposition}{{\bf \sc Proposition}}

\def\eproof{\hbox{\hskip3pt\vrule width4pt height8pt depth1.5pt}}
\def\E{{\mathrm{E}}}

\allowdisplaybreaks

\begin{document}
%\begin{comment}
\title{The Fragility of Social Learning with Noisy Messages}
\author{Matthew O. Jackson, Suraj Malladi, and David McAdams \thanks{Jackson is from the Department of Economics, Stanford University, Stanford, California 94305-6072 USA, 			and is also an external faculty member at the Santa Fe Institute. Malladi is from the Kellogg MEDS, Northwestern University. McAdams is at the Fuqua School of Business and Economics Department, Duke University. 			%Emails: jacksonm@stanford.edu, suraj.malladi@kellogg.northwestern.edu, david.mcadams@duke.edu.
We gratefully acknowledge financial support under NSF grant SES-1629446 and from Microsoft Research New England. We thank Arun Chandrasekhar, Ben Golub, Sudipta Sarangi, Omer Tamuz, and Moritz Meyer-ter-Vehn for helpful conversations and suggestions. }}
\date{Draft: August 2026}
\maketitle

\begin{abstract}
%[[100 Word version:]]
%We examine how well agents learn when information reaches them through chains of noisy person-to-person relay. If noise only takes the
%form of random mutations and transmission failures, then there is a sharp threshold such that a receiver learns fully if she has access to more chains than the threshold and nothing
%with fewer. Moreover, simple information processing rules can perform as well as fully Bayesian learning. However, if  some
%agents deliberately distort message content, learning may be impossible with any number of chains, even if the fraction of such biased individuals is small.

We examine how agents learn when information from original sources only reaches them after noisy relay.
A receiver learns if and only if they have access to sufficiently many chains of noisy relay
and they perfectly understand the noise process.
However, even slight uncertainty over message mutation rates makes learning from long chains impossible, no matter how many independent sources are accessed.
%This rationalizes long-run disagreement: even agents with a common prior and access to arbitrarily many primary sources, but facing uncertain noise, have different posteriors if they are at different distances from the primary sources.


\end{abstract}
%\end{comment}

\thispagestyle{empty}

\setcounter{page}{0} \newpage

\section{Introduction}

In a simple model of social learning resembling the children's game of `telephone,' we show that non-convergence and limited learning is generically the {\sl only} outcome.
The key difference from many standard models of social learning is the introduction of noise in communication.
We show that in the face of uncertainty about noise in communication, agents face a fundamental identification problem---they
cannot distinguish patterns in signals due to fundamentals from patterns in signals due to noise.

In our model, a set of agents are connected in a network. Some of these agents, whom we call primary sources, receive exogenous and independent binary signals about a binary state of the world (e.g., is eating avocados good or bad for one's cholesterol). They noisily relay this information to their neighbors, who relay it to their neighbors, and so on.
We model noise as independent errors at each step of communication, whereby one type of message is sent or misheard as another.
We consider the perspective of an agent (the `learner') who has slight uncertainty about the communication process, e.g., about the propensity of agents to overstate the health benefits of avocados or about the chance that someone misinterprets something they hear.  The learner is at some distance from the primary sources in the network. We ask whether this agent can learn the true state.

We show that even under ideal conditions---when learners are path connected to arbitrarily (even infinitely) many independent primary sources, know  the exact distance between themselves and those sources, and carry out perfect Bayesian updating---if they face any uncertainty about the noise, then they cannot learn the state.

Noise and distance to primary sources diminish how much an agent can learn from any given number of primary sources.
Still, if the learner knew the \emph{exact} probabilities with which ``mutations'' due to communication occur, they could learn the true state from sufficiently many sources (a number growing exponentially in distance to the sources).
However, we show that if the learner has uncertainty %(say, an atomless prior with an arbitrarily narrow support)
about the {relative probability} that messages mutate in either direction, then they cannot perfectly learn the true state, no matter how many independent primary sources they hear from.
For example, suppose that a majority of the messages that the learner hears support the view that avocados are good for cholesterol. This could either be because that is indeed true or because people are more likely to exaggerate than understate what they hear about the benefits of eating avocados.
Indeed, \emph{nothing} can be learned about the state in the limit as the distance to sources grows, even with infinitely-many sources and small uncertainty about mutation rates.

As we show, this limitation on learning
stems from a basic identification failure:  the learner cannot disentangle uncertainty about the state from uncertainty about noise in the communication process.
Either could account for the patterns of communication that an agent observes.

We also discuss how this implies that learners who have a common prior but are simply located at different distances from the primary sources  have different beliefs in the limit as the number of sources goes to infinity. These failures of learning and consensus persist \textit{a fortiori} in less ideal settings with more complex networks, bounded rationality, and noisy communication.

We note that we do not microfound the noise in agents' communication in our model
because our analysis only depends on the probabilities with which mutations occur at each step of the communication process. Thus, our results hold regardless of how this noise is generated, e.g., whether mutations arise from unintentional mistakes in interpretation, intentional miscommunication, or some technological imperfection.\footnote{For some perspectives on incentives in the spread of incorrect information, see \cite{papanastasiou2020fake}, \cite*{acemoglu2010spread}, \cite*{bloch2018rumors}, and \cite*{kranton2024social}.}

\subsection{Related Literature}

Ours is not the first model of learning that can rationalize people holding different beliefs in response to similar information.  Others include biased updating (e.g., \cite*{fryerhj2019}), having optimal models
that diverge based on sample differences (e.g., \cite*{haghtalabjp2021}),
or having misperceptions of others' news access (\cite*{bowen2021learning}).
None of these have noisy communication.

\cite{mostagir2022learning} and \cite*{acemoglu2016fragility} are closest in having uncertainty about the fundamental signal process, but they focus on heterogeneity in priors, with agents each believing the other to be misspecified.
In both models, agents expect to learn the state after observing a sequence of public signals but expect that others will not.\footnote{\cite{mostagir2022learning} focus on comparing Bayesian learning to Degroot learning.
}
By contrast, the agents in our model have no misspecifications or incorrect beliefs, but
instead face an identification problem that they cannot overcome. \cite*{akbarpour2017information} also study a setting where agents share information about a state. They find that beliefs among different agents fail to converge if each share an estimate of the state but not the precision of that estimate. The mechanism for learning failure is different, as agents update beliefs in a non-Bayesian way and share the beliefs. We focus on settings where agents directly share the signals they hear, with noise, and update beliefs in a Bayesian way.

This force is reminiscent of the \textit{confounded learning} results that \cite{mclennan1984price} and \cite{smiths2000} find in the contexts of experimentation and social learning: agents are unable to tell states apart given the observable data. While confounded learning arises there for certain prior beliefs, it happens in our model for essentially all priors, once the agent is sufficiently far from the sources.

Thus, our results complement these other explanations, showing that even fully Bayesian agents with a common prior can end up with
different posterior beliefs when facing an identification problem: being unable to distinguish uncertainty about noise in the communication process from different signals about the state.
Understanding this sort of identification problem stemming from noise, which is ubiquitous in social communication settings, is necessary in order to design policies that correct learning failures.
In particular, this fundamental identification problem  has different policy implications from other frictions that preclude consensus and learning.
To overcome the identification problem, one has to either eliminate the noise (which will not happen in many settings) or educate the population about the extent of noise in the system.\footnote{\cite*{jackson2022learning} study a similar noisy communication process but focus on non-Bayesian updating. They examine optimal policies in terms of changing the network structure to alleviate the learning failures due to noisy communication. By contrast, we identify a failure of learning that happens under Bayesian learning.}

\

\section{The Base Model of Noisy Information Transmission}\label{model}

We begin by studying how noise builds up along a chain that can travel a path of length $T$ from an
original source to a Bayesian ``learner.''

Information passes by ``word of mouth.''  This can be oral, written,
via social media, etc.

There are two possible states of the world, $\omega \in \{0, 1\}$.
Let $\theta\in (0,1)$ be the prior probability that the state is 1.

A sequence of agents $\{1,2,\ldots, T\}$, referred to as a ``chain,'' successively relays a signal of
the state via word of mouth, terminating with the learner at $T\geq 1$.

We do not model what the learner does with this information, but one can think of
the learner preferring to match their action with the state.
For instance, the learner may hear from friends about whether a certain diet is good for cardiovascular health and decide whether to adopt it.\footnote{The
	learner may have information from sources outside of its network.
	If these sources are not direct, then they can be modeled as part of the network.
	Otherwise, we can think of this external information as being reflected in the prior.
	We are interested in studying what the learner is able to learn from messages conveyed
	within their network.}

A first agent in a chain, interpreted as ``a primary source,''
observes a noisy signal of the state, $s_1 \in \{0,1, \emptyset\}$.\footnote%
{We focus on a binary world to crystallize the main ideas.  Extensions to richer state spaces and signal structures are left for future research.}
That signal is transmitted with noise becoming $s_2\in \{0,1, \emptyset\}$, and so on
until signal $s_T$ reaches the learner.

The ``null signal,'' $s_{t} = \emptyset$, indicates that no signal was received,
in which case no signal is transmitted.  Another possibility is that something was received,
but that the information was sent along in some incoherent manner:  one person hears from another but cannot understand what was said and so has no information to pass along.
In particular, if agent $t\geq 1$ receives the null signal $s_{t} = \emptyset$,
then all subsequent agents (including the learner) also receive the null signal.

If agent $t\geq 1$ receives a signal $s_t \in \{0,1\}$,
then that agent passes a signal along ($s_{t+1}\neq \emptyset$)
with probability $p_1$ if $s_t = 1$, and with probability $p_0$ if $s_t = 0$.
Thus, for instance, if $p_1 > p_0$ then agents are more likely to transmit a signal if they heard a 1, and vice versa if $p_1 < p_0$.  With the remaining probabilities of $1-p_1$ and $1-p_0$, respectively, the signal is dropped and $s_{t+1}=\emptyset$.

Each time a non-null signal is transmitted, that signal mutates from 0 to 1 with probability
$\mu_{01}\in [0,1/2)$, or from 1 to 0 with probability $\mu_{10}\in [0,1/2)$; we let $M \equiv  1 - \mu_{01} - \mu_{10}$.
We focus on the case where mutation rates are less than 1/2:
signals are more likely to be transmitted faithfully than flip at each step.
Again, these mutations could be from a person deliberately changing a message to suit their
personal preference, or could be due to some misunderstanding or other
noise in communication.

Our reduced-form model of communication suffices for our study
of the capacity for receivers to learn.
We emphasize that for any microfoundation of senders' (potentially heterogeneous)
incentives, only the resulting average probabilities of mutation and message dropping
at each step matter for our analysis.
Thus, even though we do not microfound the noise in agents' communication in our model,
it does not matter whether agents make mistakes in interpretation, have incentives to intentionally distort information, or face some technological imperfection, as long as these result
in mutations of information.

In summary, if $s_{t-1}=1$, the next agent (including $t=1$) hears: $s_t=1$ with probability $p_1(1-\mu_{10})$, $s_t=0$ with probability $p_1\mu_{10}$, and $s_t=\emptyset$ with probability $1-p_1$.
Similarly, conditional on $s_{t-1}=0$: $s_t=1$ with
probability $p_0\mu_{01}$, $s_t=0$ with probability $p_0(1-\mu_{01})$,
and $s_t=\emptyset$ with probability $1-p_0$.
If $s_t=\emptyset$ for some $t$, then $s_{t+1}=\emptyset$.
This defines a $3\times 3$ Markov chain in which $\emptyset$ is an absorbing state.\footnote
{Note that the setting is stationary in that
	the initial signal $s_1$ is derived from the original state in the same way as any other $s_t$ depends on $s_{t-1}$,
	as if nature were ``agent 0'' in the chain with signal $s_0$ equal to the state.
	This assumption simplifies the expressions,
	but our analysis easily extends to allow first-signal accuracy and transmission failure rates to
	differ from subsequent ones.}

Our analysis presumes that the learner has access to
some number $n \geq 1$ of length $T$ chains of messages, relayed through an \textit{information network}.
This network is a depth $T$ directed tree, with nodes representing agents, $n$ leaves representing the primary sources, the root representing the learner, and edges representing the direction of relay. 
Each path from a leaf to the root is a chain along which messages are forwarded.  Throughout, we measure depth/distance by the number of noisy transmission steps that a signal undergoes en route to the learner.
We let $R(n, T)$ be the set of such trees, with three examples pictured in Figure \ref{fig:tree_examples}.

\begin{figure}[h!]
	\centering
	\includegraphics[scale=0.5]{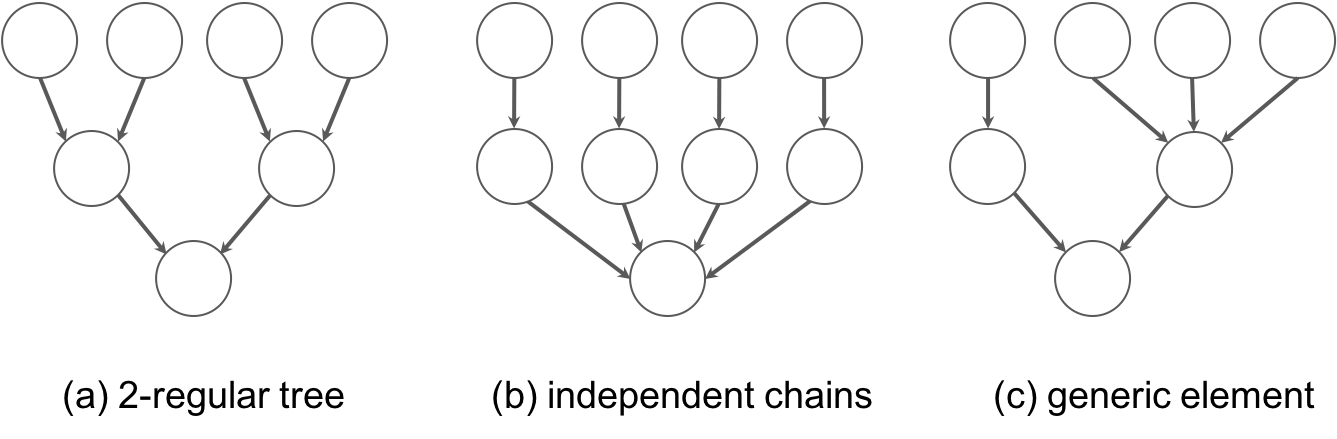}
	\caption{Three trees in $R(4, 2)$, the set of depth-2 directed trees with 4 leaves.}
	\label{fig:tree_examples}
\end{figure}

Conditionally independent signals of the state
are independently relayed along each of these chains of length $T$ via the same noisy process to the same learner.
An example of the communication process over a 2-regular, depth 4 tree is pictured in
Figure \ref{fig:comm_diagram}. The learner's ability to learn depends only on the number of primary sources, $n$, they are connected to and their distance, $T$. We  therefore leave the precise structure of the information network within $R(n, T)$ unspecified in the statement of our results, and refer to the number of sources at distance $T$, $n(T)$.

\begin{figure}[h!]
	\centering
	\includegraphics[scale=0.4]{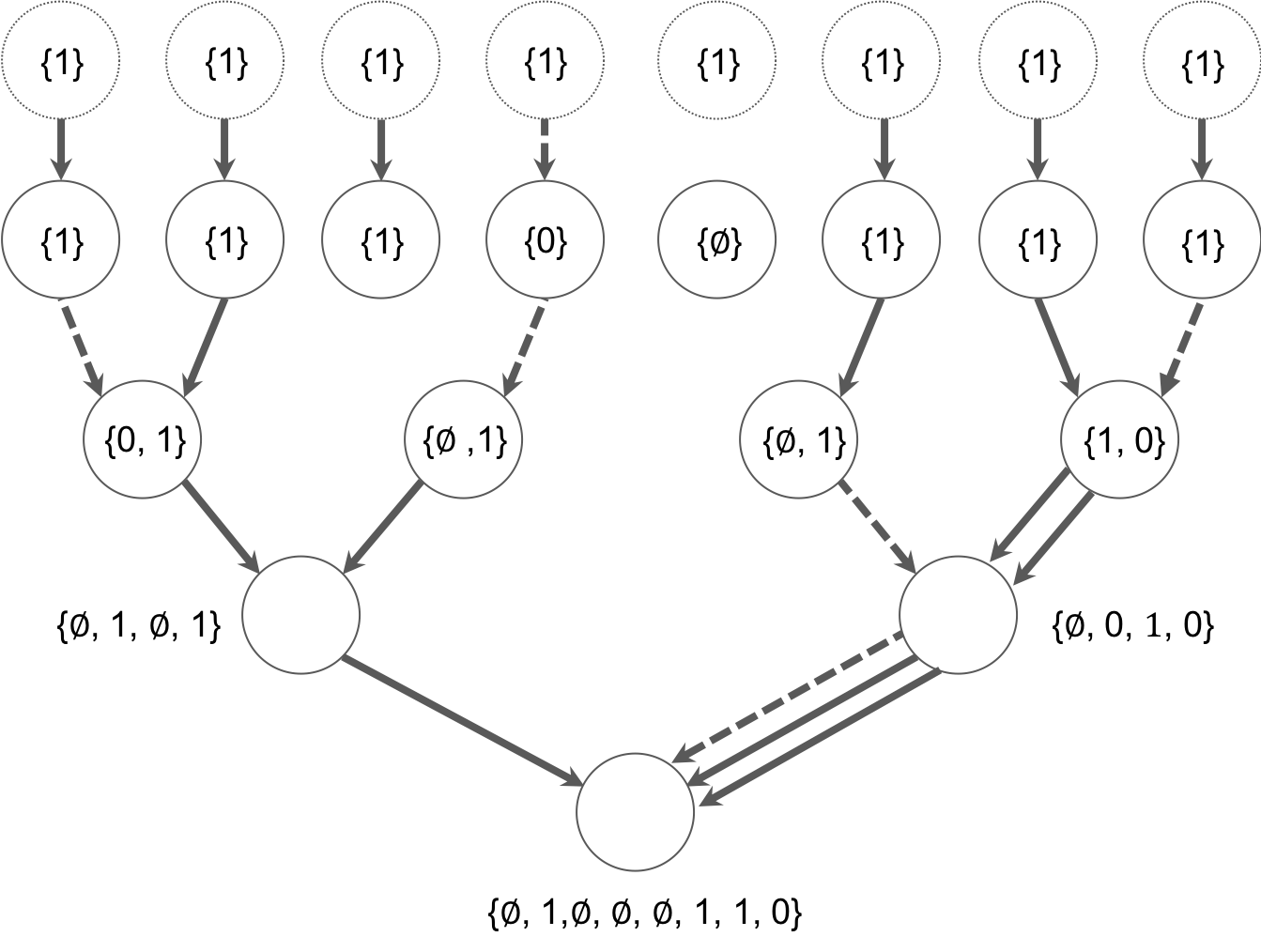}
	\caption{
		The root node (``learner'') receives messages passed through eight paths, each starting from a different source. The absence of an arrow from one node to the one below it indicates that no message was sent, a dashed arrow indicates the message was delivered but mutated, and a solid arrow indicates that the message was delivered un-mutated. In this example, the true state is $1$ and paths 1-3 and 6-8 begin with a correct initial signal, while path 4 begins with an incorrect initial signal and path 5 begins with no signal received. Initial messages are delivered on paths 1,2,4, and 6-8, mutating from 1 to 0 on path 1 and from 0 to 1 on path 4, and undelivered on paths 3 and 5. Messages are then  relayed on path 2,4, and 6-8, mutating from 1 to 0 on path 6, but dropped on path 1. Finally, messages are re-relayed on paths 2 and 6-8, mutating from 0 to 1 on path 6, but dropped on path 4. Overall, the learner hears four messages, of which two never mutated, one mutated once, and one mutated twice.}
	\label{fig:comm_diagram}
\end{figure}

Before analyzing the model, we make a few comments on its structure.

Long streams of relaying of messages and mutations are common in today's media.
For example, in \cite{adamicetal2016}'s study of online viral memes,
one meme was reposted more than 470,000 times, with a mutation rate of around 11 percent and
more than 100,000 variants.  121 of the 123 most viral
memes each had more than 100,000 variants.
There are many other examples that we discuss in \cite{jacksonmm2021}, including settings documenting extremely long strings of relayed messages
\citep{liben2008tracing}. At the same time, access to various information sources has grown as well. Accordingly, our analysis focuses on how agents learn as more messages are relayed to them, but over longer chains. Our main result is that learning is difficult when the learner has even slight uncertainty about the communication process.

In cases in which nodes relay multiple signals, one could alternatively envision nodes aggregating signals and then passing along an estimate of the state.
It is also natural to consider correlated noise across signals transmitted by any node: each node may be prone to one kind of bias.
Introducing these complications would correlate mutations making learning even more challenging.

Adding cycles in the network would do the same.
While we, as researchers, can develop expressions for Bayesian updating in tree-like networks, introducing cycles quickly makes the problem intractable.  It introduces complex correlations and expressions that explode combinatorially with the numbers of cycles and resulting walks in the network.
Experiments on networks with even just a few cycles result in subjects being much better predicted by models that involve subjects averaging signals rather than doing any Bayesian updating (e.g., see \cite*{chandrasekarlx2013}).
Thus, we focus on independent noise along trees, because at least in terms of being the only tractable setting for updating, it gives the best possibility for learning. We show that learning can be precluded even then.

\section{Learning from Message Content}\label{section_learning}

We first explore whether (and how much) the learner can learn about the true state in a special case of our model in which all messages are equally likely to be transmitted; i.e., $p_1 = p_0 = p$. In Section \ref{subsection_survival}, we allow for the possibility that message content impacts the likelihood of transmission.

Our analysis focuses on learners who use either a threshold of signal content (e.g., whether more than a given fraction of signals are 1s vs 0s) or a survival threshold (e.g., whether more than a given number of signals are observed). 

A Bayesian learner would account for both signal content and survival. However, as we show in Appendix B, the limit learning of such a Bayesian learner is the same as the better of our two simple-learning types.\footnote{We show this for the case of symmetric mutation rates. The case of asymmetric mutation rates is more complex to analyze, but we conjecture that the same equivalence holds there.} 

\subsection{Learning From Content Along a Single Chain}\label{subsection_decay}

Lemma \ref{miracle} characterizes learning along a single chain of messages when the message originates from a node at distance $t$ from the observer. Recall that $M \equiv  1 - \mu_{01} - \mu_{10}$.

\begin{lemma}
	\label{miracle}
	Suppose that $p=p_0 = p_1>0$ and consider any mutation rates $\mu_{01}, \mu_{10} \in (0,1/2)$. If the state is 0 and agent $t \geq 1$ receives a non-null message, then the message is 0 (matching the true state) with probability
	\begin{equation*}
	X_0(t) = \frac{\mu_{10} + \mu_{01}M^t}{\mu_{10} + \mu_{01}}.
	\end{equation*}
	where $M = 1 - \mu_{10} - \mu_{01}$. If the state is 1 and agent $t \geq 1$ receives a non-null message, the message is 1 (matching the true state) with probability
	\begin{equation*}
	X_1(t) = \frac{\mu_{01} + \mu_{10}M^t}{\mu_{01} + \mu_{10}}.
	\end{equation*}
	It follows that the difference between the probability of getting a given signal in its corresponding state compared to the other state is
	\begin{equation*}
	X_0(t)- (1-X_1(t)) = X_1(t)- (1-X_0(t)) = M^t.
	\end{equation*}
	As $t$ grows, regardless of the starting state, the limit probabilities that a surviving message is a 0  or a 1, respectively, are
	\begin{equation*}
	\pi_0 = \frac{\mu_{10}}{\mu_{10} + \mu_{01}}    \ \ \  \mathrm{ and } \ \ \ \pi_1 = \frac{\mu_{01}}{\mu_{10} + \mu_{01}}.
	\end{equation*}
	Finally, if $\mu_{01} = \mu_{10} = \mu$, then the message matches the true state with probability
	\begin{equation}\label{eqn_Xt}
	X(t) = \frac{1+M^t}{2}.
	\end{equation}
\end{lemma}

Note that $X(t) > 1/2$ for all $t$,  $\lim_{t \to \infty} X(t) = 1/2$, and $X(t) - 1/2$ decreases exponentially at rate $1-2\mu$.
Intuitively, the rate of decay,  $1-2\mu = (1-\mu) - \mu$, is how much more likely one is to get an unmutated signal than a mutated one from one period to the next.

\subsection{Learning from Content Along Many Chains}\label{subsection_mutation}

We now characterize the threshold number of independent
word-of-mouth chains that a Bayesian learner needs to access in order to have an accurate view of the true state.

Suppose that the learner has access to $n(T)$ independent chains of length $T$.
We index $n$ by $T$, because we seek to characterize how many chains are needed as a function of their length.  Longer chains are more likely to be null or to have an incorrect message and so more are needed to deliver an equivalent amount of information.
Let $I_{n(T)}$ be the vector of (potentially null) random messages that the learner receives from the chains, and let
the random variable $b(n(T),T)=\Pr_T (\omega=1 | I_{n(T)} )$ be the posterior probability that the state equals 1 conditional on the information from $n(T)$ originating signals that have each independently traveled $T$ steps.

%\begin{definition}[Threshold for learning]
We say that
$\tau(T)$  is a {\sl threshold for learning}  if (i) $plim \ |b(n(T),T) - \omega| = 0$ whenever $n(T)/\tau(T) \rightarrow \infty$ and (ii) $plim \  b(n(T),T) = \theta$ whenever $n(T)/\tau(T) \rightarrow 0$.\footnote{For any sequence of random variables $(X_T)_{T=1}^{\infty}$, we write $plim X_T = X$ to mean that this sequence converges to $X$ in probability.}
%\end{definition}

A threshold for learning is sharp in that if the number of chains of messages is of higher order, then the receiver learns the true state with a probability going to one, while if it is of
lower order, the receiver learns {\sl nothing}.

\begin{lemma}
	\label{multipath1}
	Consider a learner who is connected to $n(T)$ primary sources at distance $T$ and who knows $p,\mu_{01},\mu_{10}$.
	Then $\frac{1}{p^TM^{2T}}$ is a threshold for learning.
\end{lemma}

The threshold in Lemma \ref{multipath1} is sharp, and translates directly into a threshold for the average degree in the tree.
For instance,
suppose that the learner receives word-of-mouth messages
through a random tree generated by a Galton-Watson branching process  in which average degree distribution $F$ places weight 0 on degree 0 and has finite variance.\footnote%
{This condition ensures that the tree does not die out and so has at least some paths of depth $t$ with probability one. The analysis can be adapted to allow for extinction, but no new insight emerges.}
Then, one can easily show that
$plim \ b(t) = 1$ or $0$ whenever $\mathbb{E}_{d\sim F}[d]>\frac{1}{pM^{2}}$ and (ii) $plim \  b(t) = \theta$ whenever $\mathbb{E}_{d\sim F}[d]<\frac{1}{pM^{2}}$.

\subsection{The Impossibility of Learning with Uncertain Mutation Rates}\label{extremists}

Lemma \ref{multipath1} shows that learning is possible with sufficiently many primary sources, no matter how far away these sources may be, so long as the learner knows the mutation rates $\mu_{10}$ and $\mu_{01}$.
In practice, however, agents are at least somewhat uncertain about these mutation rates.
As we show next, slight uncertainty about the {ratio} of these rates can dramatically limit what can be learned even from infinitely-many sources. Moreover, how much an agent with infinitely-many sources can learn itself converges to nothing as the distance to those sources goes to infinity, meaning that learning is {\sl completely precluded} in the limit as the distance to sources grows.

Via a standard calculation (as $t$ grows) for a two-state Markov chain, the steady state limit probability that a surviving message is 1 is
\[
\frac{\mu_{01}}{\mu_{01} + \mu_{10}} \equiv \rho.
\]

\begin{proposition}
	\label{treeLearning3}

 \begin{enumerate}
\item If the learner knows $\rho$, then for any $\varepsilon>0$ and $T$ there exists $n(\varepsilon,T)$, such that if there are at least $n(\varepsilon,T)$ primary sources
then $\text{Pr}(|b(n(T), T) - \omega| < \varepsilon) > 1-\varepsilon$.

\item By contrast, if the learner does not know $\rho$ but has a prior over this ratio with a continuous and strictly positive density on a connected support, then  for any $n:\mathbb{N}\to \mathbb{N}$,
	 $b(n(T),T)$ converges in probability to $\theta$ as $T$ goes to infinity.

 \end{enumerate}
\end{proposition}

The first part of Proposition \ref{treeLearning3} says that a learner who exactly knows the relative mutation rates learns perfectly from sufficiently many sources.
The second part of Proposition \ref{treeLearning3} says that as the distance to sources grows, the learner learns nothing about the state no matter how quickly the number of sources grows.
Note that Proposition \ref{treeLearning3} makes very weak assumptions about the nature of uncertainty over $\rho$.
In particular, even if the learner has a prior with a very narrow support so that they are nearly certain about $\rho$, they cannot learn if $T$ is sufficiently large.

When networks are sufficiently shallow, small enough uncertainties
over mutation rates do not preclude learning from sufficiently many sources.
But the broad intuition for the second part of Proposition \ref{treeLearning3} is that depth compounds the effects of even small uncertainties and induces an
identification problem for the receiver: once chains are sufficiently long, full learning is no longer possible.
Observing more 1 messages than some threshold can either indicate that the
state is 1 or that people are slightly more biased towards mutating 0's to 1's
than the receiver anticipated.
The next subsection discusses this more explicitly.

The intuition for the second part of Proposition \ref{treeLearning3} is that, as $T$ grows,
the fraction of 1 versus 0 messages converges to   $\frac{\mu_{01}}{\mu_{10}}$.
Learning comes from the fact that there is a bias away from $\frac{\mu_{01}}{\mu_{10}}$ in favor of the starting state, but that bias vanishes as $T$ grows.
Lemma \ref{multipath1} says that, if there are enough signals as a function of $T$ (so $n(T)$ grows fast enough), then that slight bias can be discerned.
However, any uncertainty about $\frac{\mu_{01}}{\mu_{10}}$ completely swamps
the vanishing difference in the relative frequency of signals that reflects the starting
state.

\subsection{The Lack of Identification and the Failure of Learning}\label{app:impossibility_intuition}

To hone intuition for Proposition \ref{treeLearning3} and understand the identification failure, consider a hypothetical extreme case in which there are {\sl infinitely-many} independent chains.
In this case, the learner at distance $T$ gets an infinite number of signals that come in ratios proportional to their expected values.
We {\sl heuristically} reason at the limit to show that the results are due to identification, and not an order of limits argument.

Recall from Lemma \ref{miracle} that, when the state is 1 or 0, each surviving (non-null) message's likelihood of being 1 equals $X_1(T)$ or $1-X_0(T)$, respectively. As the number of primary sources goes to infinity, fraction $X_1(T)$ of surviving  messages are 1 if the state is 1, while fraction $1-X_0(T)$ are 1 if the state is 0. Recall also that
$$X_1(T)- (1-X_0(T)) = M^T= (1-\mu_{01}-\mu_{10})^T,$$
which implies $1-X_0(T) = X_1(T) - M^T$. Note that the fraction of 1's heard by the learner is slightly smaller when the state is 0 than when the state is 1, but only by the vanishing amount $M^T$.

Suppose that a learner at distance $T$ from primary sources sees a fraction $f$ of messages that are 1.
What are the possibilities about the state and mutation rates that can rationalize what the learner has observed?
\emph{Either} the state is 1 and $\mu_{01},\mu_{10}$ are such that
they solve
$f= X_1(T)$
\emph{or} the state is 0 and $\mu_{01},\mu_{10}$ are such that
they solve
$f = X_1(T) - M^T$.

Note that as $T$ becomes large, both $X_1(T)$ and $ X_1(T) - M^T$ converge to $\frac{\mu_{01} }{\mu_{01} + \mu_{10}}$. Therefore, the fraction of 1's heard, $f$, is almost entirely driven by the mutation rates and
nearly the same $\mu_{01},\mu_{10}$ can solve both equations.
Since the learner can easily rationalize either equation with similar mutation rates, there is no way for the learner to distinguish which state it must be---unless the learner knows $\mu_{01},\mu_{10}$ sufficiently precisely to rule out very similar combinations of these mutation rates.
As $T$ grows, the required precision converges to requiring the learner know the ratio of mutation rates exactly.

In contrast, for small $T$, so that the learner is close to the original sources, $M^T$ is nontrivial and the mutation rates that could solve the two equations are substantially different and some might be ruled out by the learner's prior.  Thus, at least partial learning can be possible in the face of uncertain mutation rates if the learner is sufficiently close to the sources.
However, given that $M^T$ decays exponentially and that information is often relayed nontrivial distances, this can be demanding.

We note that being connected to primary sources at different distances, and being able to track those distances, can improve learning.  If the learner could observe many messages at distance $t$ and then again at distance $t+1$, {\sl and can distinguish how far a message has traveled}, then they could see how messages change with additional distance. This would help her further identify $\mu_{01},\mu_{10}$. Thus, at least partial learning in the face of uncertainty is possible if the learner is \emph{either} close to all sources \emph{or} can precisely identify the distance that messages have traveled, both of which are demanding conditions.

\subsection{Rational Disagreement}

\cite{golub2017learning} conclude a review of social-learning models by noting: ``Long-run consensus is a central finding throughout this literature\ldots The consistency of this finding may cause some discomfort because we often observe disagreement empirically, even about matters of fact\ldots A theory explaining long-run disagreement, especially one with rational foundations and appropriate sensitivity to network structure, would constitute a valuable contribution.''
Such long-run disagreement is a main implication of the identification problem that we document. Our theory shows why agents with varying proximity to primary sources {\sl must} entertain
different beliefs in the face of uncertain noise, with disagreement persisting even as the number of such primary sources grows large.

We have so far taken the perspective of a single learner who is some distance away from the primary sources. However, we may also consider the updating problem of multiple learners who occupy different positions in the information network. Suppose they start with a common prior about the state. If these learners sit at different distances from primary sources (e.g., a health scientist may sit closer to people doing research on the consequences of eating avocados than a typical person), so too does their ability to learn.

As an extreme example, consider a network that has such high degree across all nodes that, absent any uncertainty over mutation rates, everyone would be able to learn the state with high certainty. Introducing a small uncertainty over the mutation rates makes it so that people sufficiently close to the sources continue to learn perfectly, while sufficiently distant nodes update very little.
Even though all agents are Bayesian and start with a common prior, they end up with different beliefs based on their network position. Uncertainty over noise can therefore rationalize long-run disagreement between path-connected agents despite extensive communication.\footnote{Even if two agents discuss their beliefs, with any noise in that communication, common knowledge implications will fail to hold allowing for consensus failure.}

\section{Learning from Message Survival}\label{subsection_survival}

Proposition \ref{treeLearning3} shows how even slight uncertainty about relative mutation rates can preclude learning, in the special case of our model in which all messages are equally likely to be transmitted at each step. But what if the content of a message affects its likelihood of being retransmitted; i.e., what if $p_1 \neq p_0$? In that context, receivers can learn not only from the content of received messages but also from how often they have heard about a given issue.
For example, it is often observed that statistically significant results are more likely to be shared and published than insignificant ones. If researchers in an academic field know that many people are working on a  topic, but do not hear of many results, then they may infer
that most results were insignificant, even without paying attention to the content of the studies that were published and shared.

We show in this section that, although agents can learn from message survival in this context (Sections \ref{app:single_chain_survival}-\ref{app:many_chains_survival}), learning from messages received over sufficiently long chains remains impossible so long as agents have any uncertainty about the relative mutation rates (Section \ref{sec:impossibility_general}).
Finally, we consider how well boundedly rational learners who pay attention to either the number of messages they have heard (ignoring the content) or instead simply to the relative fraction of  1s to 0s compare in efficiency to a Bayesian learner who pays attention to both dimensions (Section \ref{sec:bdd_learning}).
Without loss of generality, we focus on the case in which $p_1 > p_0$, meaning that people are more likely to pass along signal 1 than 0.

\subsection{Learning from Survival Along a Single Chain}\label{app:single_chain_survival}

We first characterize learning along a single chain of messages when
learning purely from survival -- where $\mu_{01} = \mu_{10}=\mu$  -- to build intuition as
to how this compares to learning purely from
content (the $p=p_0 = p_1$ case of Lemma \ref{miracle}).

In the case where only message content is informative ($p_0 = p_1$),
the content of a single message becomes nearly meaningless as chains grow long
(due to mutation). In contrast, when $p_0\neq p_1$, message survival continues to be informative
about the true state of the world even as the chain of messages grows long -- although survival
becomes decreasingly likely.

Let
\[ z  \equiv \frac{p_1}{p_0}\left( 1 + (1-2\mu)  \frac{(p_1-p_0)}{ p_0 + \mu(p_1-p_0)}\right).
\]
Note that if $p_1>p_0$ and $0\leq \mu \leq 1/2$, then  $z\geq \frac{p_1}{p_0}>1$.  % is strictly greater than 1.

\begin{lemma}\label{facts}
	Suppose that
	$1 \geq p_1 > p_0 > 0$ and  $\mu_{01}=\mu_{10}=\mu \in (0,1/2]$.\footnote%
	{If $\mu=0$ then $ \frac{Pr(s_t \neq \emptyset | \omega=1 )}{Pr(s_t \neq \emptyset | \omega=0 )}=(p_1/p_0)^t$, which
		diverges, and the problem becomes trivial.  Similarly, if $p_0=0$ then $Pr(s_t \neq \emptyset | \omega=0 ) =0$ and the problem becomes trivial. Note that here we do not require that $\mu<1/2$ since survival at the first step contains information, even if subsequent steps are completely random.  This contrasts with the case in which $p_1=p_0$, in which
		learning is precluded when $\mu=1/2$.}
	\begin{enumerate}
		\item The relative probability of message survival over a chain of
		length $t$ conditional on state 1 versus state 0 equals $\frac{p_1}{p_0}$ at $t=1$ and is uniformly bounded below by $z$ and hence away from $\frac{p_1}{p_0}$ thereafter:
		\begin{equation}\label{survivala}
		\frac{Pr(s_t \neq \emptyset | \omega=1 )}{Pr(s_t \neq \emptyset | \omega=0 )} \geq  z \geq  \frac{p_1}{p_0} \text{ for all } t > 1,
		\end{equation}
		with $z > \frac{p_1}{p_0}$ when $\mu<1/2$.
		\item The ratio in (\ref{survivala}) converges as chain-length grows:  $y\equiv \lim_{t\to \infty} \frac{Pr(s_t \neq \emptyset | \omega=1 )}{Pr(s_t \neq \emptyset | \omega=0 )}$ exists.
		\item Upon seeing a surviving message, the learner's updated belief  $Pr(\omega =1  | s_t \neq \emptyset)$ is uniformly bounded below by $\frac{\theta }{\theta + (1-\theta)/z} >  \theta$ for all $t>1$,  and bounded above in the limit by $\frac{\theta }{\theta + (1-\theta)/y}<1$.\label{c1}
		\item In the limit, updating is entirely due to signal survival
		and not content: $\lim_{t \to \infty} Pr(\omega = 1 | s_t = 1) =
		\lim_{t \to \infty} Pr(\omega = 1 | s_t = 0) = \lim_{t\to \infty}
		Pr(\omega = 1 | s_t \neq \emptyset) $.\label{p2}
	\end{enumerate}
\end{lemma}

Unlike the content of a single message, which becomes nearly meaningless as chains grow long
(due to mutation), the information conveyed by a single message's survival does not vanish in the long-chain limit.
For intuition, suppose for a moment that only the first agent in each chain was biased in favor of
message 1, and other agents transmit with probability $\widehat{p}$ regardless of signal content.
The likelihood of survival to $t$ is $p_1 \left( \widehat{p} \right)^{t-1}$ if the first agent saw signal 1 or $p_0 \left( \widehat{p} \right)^{t-1}$ if the first agent saw signal 0. Thus, the relative likelihood of survival equals $p_1/p_0 > 1$ (favoring signal 1) no matter how long the chain.  Moreover, biasing all agents in favor of transmitting message 1 further increases the relative likelihood of survival from state 1, since signal 1 is more likely to be received at each step along the chain when the true state is 1 rather than 0.

\subsection{Learning from Survival Along Many Chains}\label{app:many_chains_survival}

We next consider the challenge of learning for a receiver who only counts messages without
checking what they say. In parallel to the case of learning from signal content only ($p_0 = p_1$), the learner can discern the
state from just signal frequency as long as transmission is more likely after one signal than
the other ($p_0 \neq p_1$), there are sufficiently many starting sources of information, and the learner knows the transmission differences perfectly.

\begin{lemma}\label{learn_survival}
	Consider a learner who is connected to $n(T)$ primary sources at distance $T$. Suppose that $\mu_{01}=\mu_{10}=\mu \in ( 0,1/2]$ and $1>p_1>p_0>0$.\footnote{If $\mu=0$, then it is easy to check that the threshold is an expected degree of $1/p$, which is then the threshold for messages to survive conditional upon state $\omega=1$, which are the more likely to survive.}
	There exists $\lambda(T) = c + o(1)$ for some $c \in (0, 1)$, such that a threshold
	for learning when conditioning only upon signal survival
	is $$\frac{1}{(p_1 \lambda(T) + (1 - \lambda(T))p_0)^T}.$$
\end{lemma}

Given that messages mutate, the probability that any agent transmits a message lies somewhere between $p_1$ and $p_0$. Conditional on the initial message being 1, the overall probability that a message is transmitted all the way to the end of a length-$T$ chain must therefore take the form ${(p_1 \lambda + (1 - \lambda)p_0)^T}$ for some $\lambda \in (0,1)$. Only if the number of sequences $n(T)$ grows faster than the reciprocal of this survival probability would a growing number of signals survive, conditional on the state being 1. The learner can then discern the state (perfectly in the limit) based on the actual number of signals that survive. As with the case of learning only from content, the threshold $n(T)$ grows exponentially in chain length.

\subsection{Impossibility of Learning from Message Survival with Uncertain and Asymmetric Transmission Rates}\label{sec:impossibility_general}

The identification failure and learning-impossibility finding of Proposition \ref{treeLearning3} persists when agents are able to learn from message survival. As before,
even slight uncertainty about relative mutation rates completely precludes learning from distant sources.
To see why, note that as $T$ grows, only a vanishing fraction of chains survive.
When $p_1\neq p_0$, slightly changing $\mu_{01}/\mu_{10}$ %slightly %%DAVID NOTE: please check my movement of word "slightly", it was ambiguous
changes that fraction by orders of magnitude even though it will still be
vanishing.
This crowds out the information about
the original state that can be gleaned from survival, which dies out
over the sequence.\footnote{In particular,
	to see this (wlog) consider a case in which $1>p_1>p_0>0$.
	Let $Z_{\mu_{01}/\mu_{10}}(T,s)$ be the probability that a signal survives $T$ periods conditional on starting out as a signal $s$, given $\mu_{01}/\mu_{10}$.
	The key observation is that $Z_\pi(t,0)/ Z_{\mu_{01}/\mu_{10}-\varepsilon}(T,1)$ grows without bound as $T$ grows, for any $\varepsilon$.  Both probabilities are tending to 0, but
	eventually they mix. The probabilities are strings of products of combinations of $p_1$s and $p_0$s, and tilting that combination one way or the other eventually accumulates arbitrarily in terms of
	{\sl relative} probabilities as things are exponentiated.  Even a small shift in the fraction of mutations completely overturns the advantage of the starting state.
	Then tiny uncertainty about $\mu_{01}/\mu_{10}$ introduces much larger swings in the survival rates than the starting states.}

We remark that other forms of uncertainty can also hamper learning.
For example, agents may have uncertainty about the average degree in the network.
Even if $p_0 \neq p_1$, as distance to sources grows, survival rates of messages converge across states as messages are increasingly likely to have mutated.  Thus learning becomes increasingly dependent on
knowing the number of messages that \textit{could} have reached the learner. In that case, uncertainty over the network can preclude learning from survival. Our broader message is that when message transmission is noisy, uncertainty about the communication environment can lead to an identification problem and make learning impossible.

\subsection{Alternatives to Full Bayesian Learning}\label{sec:bdd_learning}

Our analysis dispels the mystery of non-consensus by considering the impact of noisy communication alone in an otherwise idealized model (i.e., with Bayesian learners, simple and known networks, and an abundance of communication).
Our results suggest all the more that one should expect failure of consensus in settings with bounded learners and imperfectly known and complex networks. On the other hand, this leaves open the question of what the most important bottleneck to learning might be in such settings.

One obvious possibility is that people may simply not communicate much on certain issues, so beliefs never have a chance to converge. But what is the true bottleneck to learning in the many important situations where communication is abundant? An explanation suggested by our model is that the messages agents receive over social communication channels are sufficiently uninformative that agents correctly update their beliefs relatively little and by different amounts, depending on where they sit in the communication network. An alternative possibility is that agents fail to properly update beliefs despite the informativeness of social communication.

Comparing these possibilities theoretically requires positing a model of bounded learning.
In the Appendix, we consider two forms of simple bounded learning that are natural in the context of our model. First, for any parameters of the communication and noise process, consider a Bayesian learner who faces no uncertainty about those parameters.
We show that, for this Bayesian learner, the number of primary sources needed for learning as distance grows large is the same as for the better of two types of bounded learners who only pay attention to either (i) the average content of the messages received or (ii) the number of messages received. This suggests that imperfect updating need not be a bottleneck for learning in settings like ours with abundant messages.

\section{Concluding Remarks}

We introduced a benchmark model of social learning via relayed signals in the presence of mutations and transmission failures.
We showed that, even with a perfect understanding of the transmission process, learning is challenging in that it requires an exponentially growing
number of original sources as the length of the chains over which information is relayed grows.
Moreover, the slightest uncertainty over relative mutation rates leads to an identification problem that renders learning from distant sources impossible regardless of the number of chains observed.
It also shows that reducing noise and/or reducing uncertainty about the structure of that noise are needed to lead to even the possibility of learning (and thus a consensus).\footnote{The difficulty of learning from distant sources naturally motivates learners to seek out information from closer, trusted contacts, and to down-weight or ignore more distantly-sourced information.
\cite{jackson2022learning} explore policies that can help agents learn from closer sources in another paper.}
Given that models of social learning have traditionally ignored noise in the communication structure, our conclusions suggest that some of the results on social learning be reconsidered in this light.

\bibliographystyle{aer}
\bibliography{noisycommunication}

\section*{Appendix A: Proofs}

\noindent {\bf Proof of Lemma \ref{miracle}:}
We derive the expressions of $X_0,X_1$, which can also be deduced from standard Markov chain results, but it may be useful for the reader to see the derivation.
The proof is by induction. We give the proof for $X_0$, when the state is 0. The proof for $X_1$ is symmetric and the expression for $X$ is a special case.

First, note that if $t=1$ then this expression simplifies to $1-\mu_{01}$, which is exactly the probability that the message has not mutated, and so this holds for $t=1$.

Then for the induction step, supposing that the claimed expression is correct for $t-1$, we show it is correct for $t$.

The probability of matching the true state at $t$ is the probability of not matching at $t-1$ times $\mu_{10}$ plus the probability of matching at $t-1$ times $1-\mu_{01}$,
which by the induction assumption can be written as
\begin{align*}
& \ \ \ \left[1-\frac{\mu_{10} + \mu_{01}M^{t-1}}{\mu_{10} + \mu_{01}}\right] \mu_{10}
+
\left[\frac{\mu_{10} + \mu_{01}M^{t-1}}{\mu_{10} + \mu_{01}}\right](1- \mu_{01})\\
&= \mu_{10} + \left[\frac{-\mu_{10}^2 - \mu_{01} \mu_{10} M^{t-1} + \mu_{10} - \mu_{10}\mu_{01} + \mu_{01}M^{t-1} - \mu_{01}^2M^{t-1}}{\mu_{10} + \mu_{01}}\right]\\
&= \mu_{10} + \left[\frac{-\mu_{10}^2 + \mu_{10} - \mu_{10}\mu_{01} + \mu_{01} M^{t-1}(1 -\mu_{10} - \mu_{01})}{\mu_{10} + \mu_{01}}\right]\\
&= \frac{\mu_{10}  + \mu_{01} M^t}{\mu_{10} + \mu_{01}}.
\end{align*}
as claimed.\eproof

\medskip

\noindent {\bf Proof of  Lemma \ref{multipath1}:}

Consider the case in which the true state is 0, as the other case is analogous.

The expected number of 0 messages is $n(t) p^t X_0(t)$ when the state is 0, while the expected number of  0 messages when the state is 1 is $n(t) p^t (1- X_1(t))$. The  difference in the expected number of 0 messages across states is
\[
D(t) \equiv n(t) p^t X_0(t)-n(t) p^t (1- X_1(t))= n(t) p^t M^t.
\]
If the standard deviation of the number of 0 messages in both states divided by $D(t)$ goes to 0, then by Chebychev's inequality, the probability of seeing more than $n(t) p^t (1- X_1(t)) + \frac{D(t)}{2}$ 0 messages when the state is 1 goes to zero. On the other hand, the probability of seeing fewer than this many 0 messages when the state is 0 goes to zero.

For the direction in which learning obtains, it is therefore enough to show that the ratio of standard deviation to $D(t)$ goes to zero when $n(t)$ grows faster than the threshold. Now, the standard deviation of the number of 0 messages in the 0 state divided by the amount above is
\[
\frac{ ( X_0(1-p^tX_0)n(t) p^t)^{1/2}}{n(t) p^t M^t} =
\frac{ ( X_0(1-p^tX_0))^{1/2}}{(n(t) p^t)^{1/2} M^t}.
\]

Note that the numerator converges to a constant, and so this expression goes to 0 whenever
$(n(t) p^t)^{1/2} M^t$ goes to infinity, which holds exactly when $n(t)$ grows faster than the threshold $\frac{1}{p^tM^{2t}}$.

The expressions for all the other standard deviations and differences are analogous.

For the converse, suppose that $n(t)$ grows slower than the threshold, so that $n(t)p^tM^{2t} \to 0$. Note that the Kullback--Leibler divergence between the distributions of a single chain's message under the two states is $O(p^t M^{2t})$: the message survives with probability $p^t$ under either state, and, conditional on survival, the two state-contingent Bernoulli message distributions have parameters that are bounded away from 0 and 1 and differ by $M^t$. Since the messages arriving along the $n(t)$ chains are independent, the KL divergence between the two state-contingent distributions of the learner's full information is $O(n(t)p^tM^{2t})$, which converges to zero. By Pinsker's inequality, the total variation distance between the two state-contingent distributions then converges to zero as well. The likelihoods are thus asymptotically indistinguishable, and the posterior converges in probability to the prior.\eproof

\medskip

\noindent {\bf Proof of  Proposition \ref{treeLearning3}:}

We give the proof for the case in which $p^t n(t) \rightarrow \infty$.
(With fewer paths there are even fewer signals from which to learn.)

Via a standard calculation (as $t$ grows) for a two-state Markov chain, the steady state limit probability that a surviving message is 1 is
\[
\frac{\mu_{01}}{\mu_{01} + \mu_{10}} \equiv \rho.
\]

We proceed in 2 steps:

First, we derive expressions for the probability of seeing a given set of messages conditional on the state. We use these to show that if
$\rho$ is known then Bayesian posteriors conditional on sufficiently many observations converge to the true state.

%Next, when the relative mutation rates are unknown, we derive expressions for the relative mutation rate that maximizes the likelihood of the data if the state is 1 and the rate that maximizes the likelihood when the state is 0.

Second, we use these expressions to derive the posterior for the case in which $\rho$ is unknown.   We show that regardless of the number of observations, these posteriors converge to the prior as $t$ grows.

%Finally, we argue that these likelihood-maximizing rates converge to each other as T grows large. For large T, the posteriors on relative mutation rate conditional on state concentrate around these points as well, so the states are not distinguishable given the observed signals.

The following straightforward lemma (proof omitted) is useful.

\begin{lemma}
	\label{lembinom}
	Consider a sequence of integer pairs $(k,m)$ such that $k \leq m$, $m\rightarrow \infty$, and $\frac{k}{m}\rightarrow a$.\footnote{More precisely, consider a sequence $\{(k_j,m_j):j-1,2,...\}$ with $k_j \leq m_j$ for all $j$, $\lim_{j \to \infty} m_j = \infty$, and $\lim_{j \to \infty} \frac{k_j}{m_j} = a$. We omit sequence notation in the text to simplify expressions.}
	The maximizer of $z^k (1-z)^{m-k}$ is $z(m,k)=\frac{k}{m}$,
	\[
	\frac{ z(m,k)^k (1-z(m,k))^{m-k}}{z^k (1-z)^{m-k}} \rightarrow \infty
	\]
	for any $z\neq a$, and the size of the above displayed ratio increases with the distance of $z$ from $a$ (as $z^a (1-z)^{1-a}$ is strictly concave).
	Moreover, for any atomless and continuous probability measure $G$ on $z$ that has connected support and includes $a$ in its interior
	\[
	\frac{\int_{a-\varepsilon}^{a+\varepsilon}  z^k (1-z)^{m-k}dG(z) }{\int_0^1  z^k (1-z)^{m-k} dG(z) } \rightarrow 1,
	\]
	for any $\varepsilon>0$.
\end{lemma}

Conditional on survival, the probability that some sequence starts in state $\omega=1$ and ends in a 1 equals
\[\rho + (1-\rho)M^t.\]

Conditional on survival, the probability that some sequence starts in state $\omega=0$ and ends in a 1 equals
\[\rho(1 - M^t).\]

Similar calculations provide the conditional probabilities associated with ending in a 0.

The chance of observing $k$ 1s conditional on $m$ sequences reaching the receiver and on the starting state being $\omega=1$  is then
\[
P_{k,m,t,\rho}(1) =
\binom{m}{k} \left[\rho + (1-\rho)M^t\right]^k   \left[(1-\rho) (1-M^t)\right]^{m-k}.
\]
Analogously,
the chance of observing $k$ 1s conditional on $m$ sequences reaching the receiver and on the starting state being $\omega=0$  is
\[
P_{k,m,t,\rho}(0) =
\binom{m}{k} \left[\rho(1 - M^t)\right]^k   \left[(1-\rho) + \rho M^t\right]^{m-k}.
\]

First consider the case where $\rho$ is known.

Suppose the state is 1 (the argument for the case where the state is 0 is analogous). As the number of signals grows large (keeping $t$ fixed), $\frac{k}{m-k} \to \frac{\mu_{01} + \mu_{10}M^t}{\mu_{10} - \mu_{10}M^t} \equiv a_{t,1}$ in probability. Suppose without loss of generality that $\mu_{01}\geq \mu_{10}$ so $a_{t,1}>1$ for all $t$.

The Bayesian's posterior that the state is $\omega=1$ conditional upon seeing $k$ 1's out of $m$ sequences that reached the receiver is
\[
\frac{\theta P_{k,m,t, \rho}(1)}{\theta P_{k,m,t, \rho}(1)+ (1-\theta) P_{k,m,t, \rho}(0)}.
\]
Let $k_n$ and $m_n$ be the random number of 1's and messages received respectively with $n$ length $t$ chains. By Lemma \ref{lembinom}, $\frac{P_{k_n,m_n,t, \rho}(0)}{P_{k_n,m_n,t, \rho}(1)} \rightarrow 0$ in probability as the number of signals $n$ grow large. It follows that
\[
\frac{\theta P_{k_n,m_n,t, \rho}(1)}{\theta P_{k_n,m_n,t, \rho}(1)+ (1-\theta) P_{k_n,m_n,t, \rho}(0)}\rightarrow 1,
\]
and so an agent who knows $\rho$ can learn the true state with sufficiently many signals for any given path-length $t$.
Therefore, since the agent can learn the true state with sufficiently many signals for any given $t$, it follows that the agent can learn the true state as $t\to \infty$ if $n(t)$ grows quickly enough.

Next, to conclude the proof we consider the case when $\rho$ is unknown but follows a distribution $F$ with connected support that admits a continuous and strictly positive density $f$ on that support.
We show that the Bayesian's posterior converges to $\theta$ regardless of the state as $t$ becomes large, regardless of $n(t)$.

A Bayesian's posterior that the state is $\omega=1$ conditional upon seeing $k$ 1's out of $m$ sequences that reached the receiver is
\[
\frac{\theta \int_{\rho} P_{k,m,t,\rho }(1)dF(\rho)}{\theta \int_{\rho} P_{k,m,t,\rho}(1)dF(\rho)+ (1-\theta) \int_{\rho} P_{k,m,t,\rho }(0)dF(\rho)},
\]
To complete the proof, it suffices to show that
$ \int_{\rho} P_{k,m,t,\rho }(1)dF(\rho) / \int_{\rho} P_{k,m,t,\rho}(0)dF(\rho)$ converges to one in probability.

Given a true ${\mu_{01}}^*, {\mu_{10}^*}$ such that $\rho^* = \frac{\mu_{01}^*}{\mu_{01}^* + \mu_{10}^*}$ lies in the interior of the support of $F$, and conditioning on the state being $\omega = 1$ (the case $\omega = 0$ is analogous), the realized $k,m$ will be such that
$
\frac{k}{m-k} - \frac{\mu_{01}^* + \mu_{10}^*M^t}{\mu_{10}^* - \mu_{10}^*M^t} = \frac{k}{m-k} -\frac{\rho^* + (1-\rho^*)M^t}{(1-\rho^*) - (1-\rho^*)M^t}
$ converges to 0 in probability,
and $\frac{\rho^* + (1-\rho^*)M^t}{(1-\rho^*) - (1-\rho^*)M^t} \rightarrow \frac{\rho^* }{ 1-\rho^*} =a$.

By the first part of Lemma \ref{lembinom}, for any $k,m$,
$P_{k,m,t,\rho}(1)$ is maximized when $\rho$ equals $\rho(t, k, m, 1)$ such that
\[
\rho(t, k, m, 1) + (1 - \rho(t, k, m, 1))M^t = \frac{k}{m},
\]
and $P_{k,m,t,\rho}(0)$ is maximized when $\rho$ equals $\rho(t, k, m, 0)$ such that
\[\rho(t, k, m, 0) (1- M^t) = \frac{k}{m}.
\]

$\rho(t, k, m, 1)$ and $\rho(t, k, m, 0)$ converge to each other and to $\rho^*$ in probability. It therefore follows from Lemma \ref{lembinom} that
\[
plim \ \frac{\int_{\rho} P_{k,m,t,\rho}(1)dF(\rho)}{ \int_{\rho} P_{k,m,t,\rho}(0)dF(\rho)} = plim \
\frac{\int_{\rho(t, k, m, 1)-\varepsilon}^{\rho(t, k, m, 1)+\varepsilon} P_{k,m,t,\rho}(1)dF(\rho)} {\int_{\rho(t, k, m, 0)-\varepsilon}^{\rho(t, k, m, 0)+\varepsilon} P_{k,m,t,\rho}(0)dF(\rho)}
\]
for any $\varepsilon > 0$.

Letting $[l,h]$ be the support of $\rho$, note that the two state-contingent probabilities of a surviving message being a 1, $\rho + (1-\rho)M^t$ and $\rho(1 - M^t)$, are both affine in $\rho$ with slope $1-M^t$ and differ by $M^t$. Hence
\[P_{k,m,t,\rho}(1) = P_{k,m,t,\rho + \frac{M^t}{1-M^t}}(0)\]
for all $\rho$; in particular, $\rho(t, k, m, 0) - \rho(t, k, m, 1) = \frac{M^t}{1-M^t}$, and

\[P_{k,m,t,\rho(t, k, m, 1)+\delta}(1) = P_{k,m,t,\rho(t, k, m, 0)+\delta}(0)\]
for any $\delta \in \mathbb{R}$ such that both $\rho(t, k, m, 1)+\delta$ and $\rho(t, k, m, 0)+\delta$ fall in $(l, h)$. In particular, if we let $\varepsilon_t = \frac{1}{2}\min\{\rho(t, k, m, 1) - l, h - \rho(t, k, m, 0)\}$ and if $\varepsilon_t > 0$, the intervals $[\rho(t, k, m, 1) - \varepsilon_t, \rho(t, k, m, 1) + \varepsilon_t]$ and $[\rho(t, k, m, 0) - \varepsilon_t, \rho(t, k, m, 0) + \varepsilon_t]$ strictly lie in $(l, h)$. So by the earlier observation, applying the change of variables $\rho \mapsto \rho + \rho(t, k, m, 0) - \rho(t, k, m, 1)$,
\[\int_{\rho(t, k, m, 0)-\varepsilon_t}^{\rho(t, k, m, 0)+\varepsilon_t} P_{k,m,t,\rho}(0)f(\rho)d\rho = \int_{\rho(t, k, m, 1)-\varepsilon_t}^{\rho(t, k, m, 1)+\varepsilon_t} P_{k,m,t,\rho}(1)f\big(\rho + \rho(t, k, m, 0) - \rho(t, k, m, 1)\big)d\rho.\]
The right-hand side differs from $\int_{\rho(t, k, m, 1)-\varepsilon_t}^{\rho(t, k, m, 1)+\varepsilon_t} P_{k,m,t,\rho}(1)f(\rho)d\rho$ only in that the density is evaluated at points shifted by $\rho(t, k, m, 0) - \rho(t, k, m, 1) = \frac{M^t}{1-M^t}$, which converges to 0. Since $f$ is continuous and strictly positive on the compact support $[l,h]$ (so $f$ is uniformly continuous and bounded away from zero there), the ratio of the two integrals converges to 1 in probability:
\[plim \ \frac{\int_{\rho(t, k, m, 1)-\varepsilon_t}^{\rho(t, k, m, 1)+\varepsilon_t} P_{k,m,t,\rho}(1)dF(\rho)} {\int_{\rho(t, k, m, 0)-\varepsilon_t}^{\rho(t, k, m, 0)+\varepsilon_t} P_{k,m,t,\rho}(0)dF(\rho)}=1.\]
Moreover $plim \ \varepsilon_t = \frac{1}{2}\min\{\rho^* - l, h - \rho^*\} > 0$. Therefore,
\[plim \ \frac{\int_{\rho} P_{k,m,t,\rho}(1)dF(\rho)}{ \int_{\rho} P_{k,m,t,\rho}(0)dF(\rho)} = plim \
\frac{\int_{\rho(t, k, m, 1)-\varepsilon_t}^{\rho(t, k, m, 1)+\varepsilon_t} P_{k,m,t,\rho}(1)dF(\rho)} {\int_{\rho(t, k, m, 0)-\varepsilon_t}^{\rho(t, k, m, 0)+\varepsilon_t} P_{k,m,t,\rho}(0)dF(\rho)}=1,\]
which concludes the proof.\eproof

\medskip

\noindent {\bf Proof of Lemma \ref{facts}, Part 1:}

For ease of notation,
let $P^t_{1S} \equiv Pr(s_t \neq \emptyset | \omega=1 )$
and $P^t_{0S} \equiv Pr(s_t \neq \emptyset | \omega=0 )$.
These are the probabilities of signal survival to time $t$ conditional on the first signal.

First we prove that
$\frac{P^t_{1S}}{ P^t_{0S}}\geq \frac{p_1}{p_0} $, with strict inequality when $\mu<1/2$ and $t>1$.

This is proven by induction.
First, $P^1_{1S} = p_1 > p_0=P^1_{0S} $.
Next, let us show that  $\frac{P^t_{1S}}{ P^t_{0S}}\geq \frac{p_1}{p_0}$ given that $P^{t-1}_{1S} >P^{t-1}_{0S} $.
Note that given the stationarity of the process, $P^{t-1}_{1S} = Pr(s_t \neq \emptyset | s_1=1)$ and $P^{t-1}_{0S} = Pr(s_t \neq \emptyset | s_1=0)$, and
then we can write\footnote{The starting state $s_0$ is 1 in this calculation and so then there is a probability $p$ that the signal survives to the first period,
	and then the calculation inside the $\left[ \cdot\right]$ handles the two possible values of the first period signal and then the probability
	the signal survives to $t$ if it has made it to the first period in the two possible values it could have in the first period.}
The first part
$$P^t_{1S} = p_1 \left[(1-\mu) Pr(s_t \neq \emptyset | s_1=1) + \mu  Pr(s_t \neq \emptyset | s_1=0)\right],$$
and so then it follows that
$$P^t_{1S} = p_1 (1-\mu) P^{t-1}_{1S} + p_1\mu  P^{t-1}_{0S}.$$
Then by the inductive step ($P^{t-1}_{1S} >P^{t-1}_{0S} $) and so it follows that
$$P^t_{1S} \geq p_1 (1-\mu) P^{t-1}_{0S} + p_1\mu  P^{t-1}_{1S},$$
with strict inequality when $\mu< 1/2$ and $t>1$.
Similarly,
$$P^t_{0S} =p_0 (1-\mu) P^{t-1}_{0S} + p_0\mu  P^{t-1}_{1S}.$$
Therefore
$$\frac{P^t_{1S}}{ P^t_{0S}} \geq \left(\frac{p_1}{p_0}\right)\frac{ (1-\mu) P^{t-1}_{0S} + \mu  P^{t-1}_{1S}}{  (1-\mu) P^{t-1}_{0S} + \mu  P^{t-1}_{1S}}=\frac{p_1}{p_0},$$
with strict inequality when $\mu< 1/2$ and $t>1$,
as claimed.

Now we complete the proof of the first part of the lemma.
Note that (from above)
$$\frac{P^t_{1S}}{ P^t_{0S}} = \left(\frac{p_1}{p_0}\right)\frac{ (1-\mu) P^{t-1}_{1S} + \mu  P^{t-1}_{0S}}{  (1-\mu) P^{t-1}_{0S} + \mu  P^{t-1}_{1S}}.$$
Therefore,
$$\frac{P^t_{1S}}{ P^t_{0S}} = \left(\frac{p_1}{p_0}\right)\left(\frac{ (1-\mu) P^{t-1}_{0S} + \mu  P^{t-1}_{1S} + (1-2\mu)(P^{t-1}_{1S}-P^{t-1}_{0S})}{  (1-\mu) P^{t-1}_{0S} + \mu  P^{t-1}_{1S}}\right),$$
and then since $p_1>p_0$ and  $\frac{P^{t-1}_{1S}}{ P^{t-1}_{0S}}\geq \frac{p_1}{p_0} $, with strict inequality when $\mu< 1/2$ and $t>1$, it follows that
$$\frac{P^t_{1S}}{ P^t_{0S}}=
\left(\frac{p_1}{p_0}\right)\left(1+ (1-2\mu)\frac{P^{t-1}_{1S}-P^{t-1}_{0S}}{ P^{t-1}_{0S}+ \mu(P^{t-1}_{1S}-P^{t-1}_{0S})}\right) \geq \frac{p_1}{p_0}\left( 1 +  (1-2\mu) \frac{(p_1-p_0)}{ p_0 + \mu(p_1-p_0)}\right),$$
for all $t> 1$, with equality at $t=2$ (and it directly follows that this expression ($z$) is strictly larger than $p_1/p_0$ when $\mu<1/2$),
as claimed.\eproof

\medskip

The following result is useful in the proofs of the remaining parts of Lemma \ref{facts}.

\begin{lemma}\label{ordering_probs}
	Fix $\theta \in (0, 1), \mu_{01} = \mu_{10} = \mu \in (0, 1/2], 0 < p_0 \leq p_1 \leq 1$. For all $t>0$,
	\[Pr(s_t = 1 | \omega = 1) \geq Pr(s_t = 0 | \omega = 1).\]
	Moreover,  either there exists $T$ large enough such that
	\[
	Pr(s_t = 1 | \omega = 0) \geq Pr(s_t = 0 | \omega = 0)
	\text{ \ for all \ } \text{ for all } t\geq T,
	\]
	or
	\[
	Pr(s_t = 1 | \omega = 0) < Pr(s_t = 0 | \omega = 0) \text{ for all } t.
	\]
	Finally, the sequence
	\[ \frac{Pr(s_t = 1 | \omega = 1)}{\min\{Pr(s_t = 1 | \omega = 0), Pr(s_t = 0 | \omega = 0)\}}\]
	is bounded above.
\end{lemma}

\noindent {\bf Proof of Lemma \ref{ordering_probs}:}

The first claim is proven by induction:

Since $\mu \leq 1/2$, $Pr(s_1 = 1 | \omega = 1) \geq Pr(s_1 = 0 | \omega = 1)$. Suppose $Pr(s_t = 1 | \omega = 1) \geq Pr(s_t = 0 | \omega = 1)$.
Note that,
\begin{align*}
Pr(s_{t+1} = 1 | \omega = 1) &= p_1 (1 - \mu) Pr(s_t = 1 | \omega = 1) + p_0 \mu Pr(s_t = 0 | \omega = 1) \\
Pr(s_{t+1} = 0 | \omega = 1) &= p_1 \mu Pr(s_t = 1 | \omega = 1) + p_0(1 -  \mu) Pr(s_t = 0 | \omega = 1).
\end{align*}
The result then follows from the inductive hypothesis and the facts that $p_1 \geq p_0$ and $\mu \leq 1/2$.

Next, to show the second claim in the lemma, note that
\begin{align*}
Pr(s_{t+1} = 1 | \omega = 0) &= p_1 (1 - \mu) Pr(s_t = 1 | \omega = 0) + p_0 \mu Pr(s_t = 0 | \omega = 0) \\
Pr(s_{t+1} = 0 | \omega = 0) &= p_1 \mu Pr(s_t = 1 | \omega = 0) + p_0(1 -  \mu) Pr(s_t = 0 | \omega = 0).
\end{align*}
Then if $Pr(s_t = 1 | \omega = 0) \geq Pr(s_t = 0 | \omega = 0)$ for some $t=T$, the same will hold for all $t>T$ by a similar inductive proof.
Otherwise $Pr(s_t = 1 | \omega = 0) < Pr(s_t = 0 | \omega = 0)$ for all $t$, and then the result holds directly.

Finally, we show the third part of the claim.
By the second part of this lemma, there are two cases to consider. If  $Pr(s_t = 1 | \omega = 0) < Pr(s_t = 0 | \omega = 0)$ for all $t$. Then
\[
\frac{Pr(s_t = 1 | \omega = 1)}{\min\{Pr(s_t = 1 | \omega = 0), Pr(s_t = 0 | \omega = 0)\}} =
\frac{Pr(s_t = 1 | \omega = 1)}{Pr(s_t = 1 | \omega = 0)}
\]
If instead there is a $T$ such that for all $t \geq T$,  $Pr(s_t = 1 | \omega = 0) \geq Pr(s_t = 0 | \omega = 0)$, then
\begin{align*}
\frac{Pr(s_t = 1 | \omega = 1)}{\min\{Pr(s_t = 1 | \omega = 0), Pr(s_t = 0 | \omega = 0)\}} &=
\frac{Pr(s_t = 1 | \omega = 1)}{Pr(s_t = 0 | \omega = 0)} \\
&= \frac{p_1(1-\mu) Pr(s_{t-1} = 1 | \omega=1) + p_0\mu Pr(s_{t-1}=0 | \omega=1)}{p_1\mu Pr(s_{t-1} = 1 | \omega=0) + p_0(1-\mu) Pr(s_{t-1}=0 | \omega=0)} \\
&\leq \frac{(p_1(1-\mu) + p_0\mu) Pr(s_{t-1} = 1 | \omega=1)}{p_1\mu Pr(s_{t-1} = 1 | \omega=0) + p_0(1-\mu) Pr(s_{t-1}=0 | \omega=0)} \\
&< \frac{p_1(1-\mu) + p_0\mu}{p_1\mu}\frac{Pr(s_{t-1} = 1 | \omega = 1)}{Pr(s_{t-1} = 1 | \omega = 0)},
\end{align*}
where the second to last inequality uses the first part of this lemma. We can therefore handle both cases simultaneously by showing that the sequence $\frac{Pr(s_{t} = 1 | \omega = 1)}{Pr(s_{t} = 1 | \omega = 0)}$ is bounded above.

To that end, note that
\begin{align*}
Pr(s_t =1|\omega=0) &\geq Pr(s_t = 1 | \omega=0, s_1 = 1) Pr(s_1 = 1|\omega=0) \\
&= Pr(s_{t-1} = 1 | \omega=1) p_0\mu.
\end{align*}
So,
\[
\frac{Pr(s_{t} = 1 | \omega = 1)}{Pr(s_{t} = 1 | \omega = 0)} \leq \frac{Pr(s_{t} = 1 | \omega = 1)}{Pr(s_{t-1} = 1 | \omega = 1)}\frac{1}{p_0\mu}.
\]
It then suffices to show that ${Pr(s_{t} = 1 | \omega = 1)}\leq {Pr(s_{t-1} = 1 | \omega = 1)}$, since then
from above
\[
\frac{Pr(s_{t} = 1 | \omega = 1)}{Pr(s_{t} = 1 | \omega = 0)} \leq \frac{1}{p_0\mu},
\]
which is finite given that $p_0>0$ and $\mu>0$.
To see that  ${Pr(s_{t} = 1 | \omega = 1)}\leq {Pr(s_{t-1} = 1 | \omega = 1)}$,
\begin{align*}
Pr(s_{t} = 1 | \omega = 1)
&= p_1 (1 - \mu) Pr(s_{t} = 1 | s_1 = 1) + p_1\mu Pr(s_{t} = 1 |s_1 = 0)\\
&= p_1 (1 - \mu) Pr(s_{t-1} = 1 | \omega = 1) + p_1\mu Pr(s_{t-1} = 1 | \omega = 0)\\
&\leq  p_1 (1 - \mu) Pr(s_{t-1} = 1 | \omega = 1) + p_1 \mu Pr(s_{t-1} = 1 | \omega = 1)\\
&=  p_1  Pr(s_{t-1} = 1 | \omega = 1),
\end{align*}
where the inequality follows from the first part of the lemma, establishing the claim.\eproof

\medskip

\noindent {\bf Proof of Lemma \ref{facts}, Part 2:}

We show that $\lim_{t \to \infty}\frac{P^t_{1S}}{ P^t_{0S}} = \lim_{t \to \infty} \frac{p_1 Pr(s_{t-1} = 1 | \omega=1) + p_0 Pr(s_{t-1} = 0 | \omega=1)}{p_1 Pr(s_{t-1} = 1 | \omega=0) + p_0 Pr(s_{t-1} = 0 | \omega=0)}$ exists.

The sequence is bounded above by the first and last part of Lemma \ref{ordering_probs}:
it is bounded above by either $\frac{Pr(s_{t} = 1 | \omega=1)}{Pr(s_{t} = 1 | \omega=0)}$
or $\frac{Pr(s_{t} = 1 | \omega=1)}{Pr(s_{t} = 0 | \omega=0)}$, both of which are bounded above.
Furthermore, the sequence is bounded below by the first part of Lemma \ref{facts}.

To complete the proof that the limit exists, we show that the sequence is monotone. For this, we will start by writing, $r_t$, the $t^{th}$ term in the sequence, as $\frac{Pr(s_{t-1} = 1 | \omega=1) + \ell_1 Pr(s_{t-1} = 0 | \omega=1)}{ Pr(s_{t-1} = 1 | \omega=0) + \ell_1 Pr(s_{t-1} = 0 | \omega=0)}$, where $\ell_1 = p_0/p_1$. Now the $t+1^{st}$ is
\begin{align*}
r_{t+1} &= \frac{Pr(s_{t} = 1 | \omega=1) + \ell_1 Pr(s_{t} = 0 | \omega=1)}{ Pr(s_{t} = 1 | \omega=0) + \ell_1 Pr(s_{t} = 0 | \omega=0)} \\
&= \frac{(p_1(1-\mu) + \ell_1 p_1 \mu) Pr(s_{t} = 1 | s_1=1)  +  (p_0 \mu + \ell_1 p_0(1 - \mu))Pr(s_{t} = 0 | s_1=1)}{(p_1(1-\mu) +  \ell_1 p_1\mu) Pr(s_{t} = 1 | s_1=0) + (p_0\mu + \ell_1 p_0(1 - \mu))Pr(s_{t} = 0 | s_1=0)}\\
&=  \frac{Pr(s_{t-1} = 1 | \omega=1) + \ell_2 Pr(s_{t-1} = 0 | \omega=1)}{Pr(s_{t-1} = 1 | \omega=0) + \ell_2 Pr(s_{t-1} = 0 | \omega=0)},
\end{align*}
where $\ell_2 = \frac{p_0}{p_1}\frac{\mu + \ell_1 (1 - \mu)}{(1-\mu) + \ell_1 \mu}$. Consider the sequence $\ell_t$, where $\ell_{t+1} = \frac{p_0}{p_1}\frac{\mu + \ell_t (1 - \mu)}{(1-\mu) + \ell_t \mu}$ and $\ell_1 = \frac{p_0}{p_1}$.
Note that $\ell_t$ is non-increasing in $t$ given that $\mu\leq 1/2$ and it is strictly decreasing when $\mu<1/2$.
Iterating
on the above logic
\[r_t = \frac{Pr(s_{1} = 1 | \omega=1) + \ell_{t-1} Pr(s_{1} = 0 | \omega=1)}{ Pr(s_{1} = 1 | \omega=0) + \ell_{t-1} Pr(s_{1} = 0 | \omega=0)}.\]
To see that $r_t$ is monotone in $t$, note that the sign of the derivative of $r_t$ with respect to $\ell_t$ only
depends on the sign of
$Pr(s_{1} = 0 | \omega=1)Pr(s_{1} = 1 | \omega=0) - Pr(s_{1} = 1 | \omega=1)Pr(s_{1} = 0 | \omega=0)$,
and so it is monotone given the monotonicity of $\ell_t$ in $t$.\eproof

\medskip

\noindent {\bf Proof of Lemma \ref{facts}, Part 3:}

That $Pr(\omega = 1| s_t \neq \emptyset) \geq \frac{\theta z}{1 + \theta(z - 1)}$ for any $t>1$, with strict inequality when $\mu<1/2$, follows
from Part 1 and Bayes rule (and it is evident from the proof that this lower bound is not
tight). Therefore, it remains to show that $\lim_{t \to \infty} Pr(\omega = 1| s_t \neq \emptyset)$
exists, a step which is deferred to the proof of Part 4.

The fact that $\lim_{t\to \infty} Pr(\omega =1  | s_t \neq \emptyset) = \frac{\theta }{\theta + (1-\theta)/y}<1$ follows from Part 2 and Bayes' Rule.\eproof

\medskip

\noindent {\bf Proof of Lemma \ref{facts}, Part 4:}

It suffices to show that $\lim_{t \to \infty} Pr(\omega = 1| s_t = 1) = \lim_{t \to \infty} Pr(\omega = 1| s_t = 0)$, as this implies that $\lim_{t \to \infty} Pr(\omega = 1| s_t \neq \emptyset)$ exists and has the same value. This limiting equality between posterior distributions can equivalently be expressed in terms of likelihood ratios:
\begin{align*}
\lim_{t \to \infty} \frac{Pr(s_t=1 | \omega=0)}{Pr(s_t=1 | \omega=1)} &= \lim_{t \to \infty} \frac{Pr(s_t=0 | \omega=0)}{Pr(s_t=0 | \omega=1)} \\
\iff \lim_{t \to \infty} \frac{Pr(s_t=1 | s_t \neq \emptyset, \omega=0)}{Pr(s_t=1 | s_t \neq \emptyset, \omega=1)} &= \lim_{t \to \infty} \frac{Pr(s_t=0 | s_t \neq \emptyset, \omega=0)}{Pr(s_t=0 | s_t \neq \emptyset, \omega=1)}.\addtocounter{equation}{1}\tag{\theequation} \label{nts}
\end{align*}
We show that
\begin{equation}\label{nts1}
\lim_{t \to \infty} Pr(s_t=1 | s_t \neq \emptyset, \omega=0) = \lim_{t \to \infty} Pr(s_t=1 | s_t \neq \emptyset, \omega=1),
\end{equation}
since this implies that both sides of equation \ref{nts} are equal to 1.\footnote{Subtract each side of equation \ref{nts1} from 1 before taking ratios to see that the right side of equation \ref{nts} is also 1.}

Denote by $S$ a sequence of signals that evolve according to our process, starting with $S_0=1$ and $S'$ another (independent) sequence of signals with $S'_0=0$. Let $\tau = \min \{t | S'_t = 1\}$, where $\tau = \infty$ if $S'$ is dropped at some step before mutating to signal $1$, or if $S'_t = 0$ for all $t$.

In this notation, equation \ref{nts1} can equivalently be expressed as:  $\lim_{t \to \infty} Pr(S_t = 1 | S_t \neq \emptyset) = \lim_{t \to \infty} Pr(S'_t = 1 | S'_t \neq \emptyset)$.
Note the following relationship between the two independent paths:\footnote{Note that if $\tau>t$, then the probability that $S'_t=1$ is 0.}
\begin{align*}
Pr(S'_t = 1 | S'_t \neq \emptyset) &= \sum_{i=1}^{t} Pr(S'_t = 1 | S'_t \neq \emptyset, \tau = i) Pr(\tau = i | S'_t \neq \emptyset) \\
&=  \sum_{i=1}^{t} Pr(S_{t-i} = 1 | S_{t-i} \neq \emptyset) Pr(\tau = i |S'_t \neq \emptyset) \\
&\equiv \left( \sum_{i=1}^{t} Pr(S_{t-i} = 1 | S_{t-i} \neq \emptyset) w_i^t \right)
\addtocounter{equation}{1}\tag{\theequation} \label{sumexpanded},
\end{align*}
where $w_i^t = Pr(\tau = i |S'_t \neq \emptyset)$.

The result then follows from the following three claims, to be proved:
\begin{enumerate}
	\item For any $\varepsilon > 0$ and positive integer $k$, for all sufficiently large $t$, $\sum_{i=t-k}^{t} w^t_i < \varepsilon$.\label{fact1}
	\item $\lim_{t \to \infty} Pr(S_t = 1 | S_t \neq \emptyset)$ exists.\label{fact2}
	\item  Letting $w_{\infty}^t \equiv Pr(\tau > t | S'_t \neq \emptyset)$, we have $\sum_{i =1}^t w_i^t + w_{\infty}^t = 1$. Moreover, $w_{\infty}^t \to 0$ as $t \to \infty$, i.e., the probability that the signal has not yet mutated to a 1 by time $t$, conditional on survival to $t$, goes to 0 as $t$ grows.\label{fact3}
\end{enumerate}

To see that these claims imply the result, note that by claim \ref{fact1}, most of the weight falls on the first $t-k$ terms of the sum in equation  \ref{sumexpanded} for large enough $t$.  By claim \ref{fact2}, for a large enough $k$ (growing slower than $t$), these first $t-k$ terms will be close to  $\lim_{t \to \infty} Pr(S_t = 1 | S_t \neq \emptyset)$, and therefore by claim \ref{fact3} the limiting weighted sum of these terms converges to this value as well.

Claim \ref{fact3} is clear, so we prove the other two.

First we prove claim \ref{fact1}. Note that $p_0^i(1-\mu)^{i-1}\mu$ is the probability of survival with no mutation through $i-1$ and then survival with mutation at $t = i$, i.e., $Pr(\tau = i) = p_0^i(1-\mu)^{i-1}\mu$. Second, let $m_i$ be number of mutations through time $i$. Obviously, $Pr(S'_i \neq \emptyset) > Pr(S'_i \neq \emptyset$ and $m_i = 1)$. Third, if survival were always at rate $p_0$, then $Pr(S'_i \neq \emptyset$ and $m_i = 1) = i p_0^i(1-\mu)^{i-1}\mu$. However, since survival likelihood immediately after the first mutation, $p$, is strictly higher than $p_0$ and mutations sometimes occur (note, we assume $\mu > 0$), $Pr(S'_i \neq \emptyset$ and $m_i = 1) > i p_0^i(1-\mu)^{i-1}\mu$. Putting these observations together, we have
\begin{equation}\label{tau_ineq}
0 \leq \frac{Pr(\tau = i)}{Pr(S'_i \neq \emptyset)} < \frac{p_0^i(1-\mu)^{i-1}\mu}{i p_0^i(1-\mu)^{i-1}\mu}  = \frac{1}{i},
\end{equation}
so that $\frac{Pr(\tau = i)}{Pr(S'_i \neq \emptyset)} \to 0$ as $i \to \infty$, where, as noted earlier, the strict inequality arises from replacing $Pr(S'_i \neq \emptyset)$ with a lower bound on the probability of exactly one mutation occurring over the course of the first $i$ periods, and all the ways this could happen, and then
noting that $p_0<p_1$.
Now
\begin{align*}
Pr(\tau = i | S'_t \neq \emptyset) &= \frac{Pr(S'_t \neq \emptyset | \tau = i)Pr(\tau = i)}{Pr(S'_t \neq \emptyset)} \\
&= \frac{Pr(S_{t-i} \neq \emptyset)Pr(\tau = i)}{Pr(S'_t \neq \emptyset)} \\
&= \frac{Pr(S_{t-i} \neq \emptyset)Pr(\tau = i)}{Pr(S'_{t-i} = 1)Pr(S_i \neq \emptyset) + Pr(S'_{t-i} = 0)Pr(S'_i \neq \emptyset)} \\
&< \frac{Pr(S_{t-i} \neq \emptyset)}{Pr(S'_{t-i} \neq \emptyset)} \frac{Pr(\tau = i)}{Pr(S'_i \neq \emptyset)},
\end{align*}
where the inequality follows from the fact that $Pr(S_i \neq \emptyset) > Pr(S'_i \neq \emptyset)$, by Lemma \ref{facts}, Part 1. $\frac{Pr(S_{t-i} \neq \emptyset)}{Pr(S'_{t-i} \neq \emptyset)}$ is bounded by Lemma \ref{facts} Part 2 (as it has a limit), and $\frac{Pr(\tau = i)}{Pr(S'_i \neq \emptyset)}$ can be made arbitrarily small for large enough $i$ by equation \ref{tau_ineq}.
Thus, for any $\delta$ and $k$ we can find large enough $t$ for which $w_i^t< \delta$ for $i>t-k$.
Choosing $\delta=\varepsilon/k$ establishes claim \ref{fact1}.

Finally, we prove claim \ref{fact2}. The probability distribution of $S_t$ is given by $e_1^\prime A^t$, where $$ A =
\begin{bmatrix}
p_1(1 - \mu) & p_1\mu & 1-p_1 \\
p_0\mu & p_0(1-\mu) & 1-p_0\\
0 & 0 & 1
\end{bmatrix}$$ is the Markov transition matrix for $S$. Let $B$ be the principal $2\times 2$ submatrix of $A$.
By the partitioned matrix multiplication formula, $Pr(S_t = 1 | S_t \neq \emptyset) = \frac{e_1^\prime B^t e_1}{e_1^\prime B^t \mathbf{1}}.$ Since $B$ is strictly positive, the Perron-Frobenius theorem implies that this expression converges to the first coordinate of the left eigenvector corresponding to the largest eigenvalue of $B$, normalized so that its coordinates sum to one.\eproof

\medskip

\noindent {\bf Proof of Lemma \ref{learn_survival}:}

$\lim_{t \to \infty}\frac{P^t_{0S}}{ P^t_{1S} } = r$ for some $r<1$, by Lemma \ref{facts}. Let $r_t$ be the $t^{th}$ term in the sequence.

Let $m(t)$ be the number of surviving signals.
By Chernoff bounds, it follows that
\[
Pr ( m(t) > n(t) P^t_{1s} (1+r_t)/2 | \omega = 1) \rightarrow 1
\]
and
\[
Pr ( m(t) < n(t) P^t_{1s} (1+r_t)/2 | \omega = 0) \rightarrow 1
\]
provided that $n(t) P^t_{1s} \rightarrow \infty$.
Given this separation, it is easy to the check that if $n(t) P^t_{1s} \rightarrow \infty$, the beliefs will converge to 0 or 1 in probability.

Next, note that if $n(t) P^t_{1s} \rightarrow 0$, then the expected number of surviving signals in either state is 0, and that happens with the probability going to 1 by Chebychev, and so
there is no learning.   So, the threshold is $1/ P^t_{1s}$.

Note that the probability of survival lies between $p_0^t$ and $p_1^t$ (so its reciprocal lies between $1/p_1^t$ and $1/p_0^t$) and so
$$1/ P^t_{1s} = \frac{1}{(p_1 \lambda(t) + (1 - \lambda(t))p_0)^t}.$$
The fact that $\lambda(t)$ converges to some $\lambda$ then follows
from the Perron-Frobenius argument in the proof of Lemma \ref{facts}, Part 4: the distribution of the signal conditional
on survival converges to the quasi-stationary distribution given by the Perron eigenvector of the transient submatrix $B$,
and so the per-period probability of survival converges to a limit that lies between $p_0$ and $p_1$.\eproof

\section*{Appendix B: Full Bayesian Learning vs Learning Only from Survival or Only from Content}
\label{learnonly}

In this section, we provide a bound on how much more likely a Bayesian agent using both signal survival and message content is to guess the true state compared to agents who use rules of thumb that account only for signal survival or only for average signal content.

Without loss of generality, we focus on the case in which $p_1 \geq p_0$.

First, let us consider the case in which
the learner has access to a single chain and needs to predict the state based on the signal $s_t \in \{0,1,\emptyset\}$. We consider four different ways in which the learner might guess.

\begin{itemize}
\item A ``Bayesian agent,'' $B$, guesses the most likely state conditional on both signal survival and signal content.
\item A ``survival rule-of-thumb agent,''  $S$, guesses 1 if a signal is received and guesses 0 if no signal is received.
\item A ``content rule-of-thumb agent,'' $C$, guesses 1 if signal 1 is received, 0 if signal 0 is received, and guesses in favor of the prior if no message is received (flipping a coin if $\theta=1/2$).
\item A ``naive agent,''  $N$,  always guesses in favor of the prior.
\end{itemize}

 $S,C,N$ are collectively referred to as ``limited learners'' since they make their guess based on less information than is available.

\begin{proposition}
\label{relativeLearning}
Suppose that $1 \geq p_1 \geq p_0 \geq 0$ and $\mu_{01} = \mu_{10} = \mu \in [0,1/2]$.
The probability that a Bayesian agent is correct in guessing the state is at most
$\frac{4}{3}$ higher
than the best limited learner when $t=1$, and at most $\frac{3}{2}$ higher than the best limited learner for all $t>1$.\footnote{We conjecture that the bound is $\frac{4}{3}$ for any $t>1$.}
Moreover, as $t$ grows, this upper bound converges to 1.
\end{proposition}

Proposition \ref{relativeLearning} implies that, when word-of-mouth chains are long, a belief-updating strategy that uses only message survival or only message content is approximately equivalent to one that uses all available information, no matter what the parameters and no matter what the realized state.\footnote{This is obvious when $\mu=1/2$, in which case message content contains no information, or when $p_0=p_1$, in which case message survival contains no information.}

Next, suppose that the learner observes multiple chains. In this context, define ``C" to be an agent who guesses 1 whenever the fraction of 1 messages compared to 0 messages is above a threshold, and define ``S" to be the an agent who guesses 1 whenever the number of messages that survive is above or below a threshold.  These thresholds are the conditional Bayesian ones, but these agents only consider one aspect of the information available.

It is difficult to give tight bounds on the relative performance of the Bayesian agent and the best of the limited learners when there are many sequences.
However, we establish limiting results. In particular, everywhere in the parameter space, the threshold for learning for agent B is the same as for one of the limited learners.  Thus, there is no number of starting messages for which a Bayesian agent can learn but none of the naive agents can.
Indeed, for large $t$ full learning can be obtained from just one dimension, and we get the following result.

\begin{proposition}\label{multiSignalRelativeLearning}
	For any $\theta, p_0, p_1 \in [0, 1]$ and $\mu_{01} = \mu_{10} = \mu \in [0, \frac{1}{2}]$, the threshold for learning is the same for $B$ as it is for the better of $C$ or $S$.
\end{proposition}

\medskip

The next lemma is useful in the  proof of Proposition \ref{relativeLearning}.

Let $P_1^t$ ($P_0^t$) denote the Bayesian posterior probability that the state is 1 conditional upon a signal being received at time $t$ and being 1 (0).  Similarly, let $P_\emptyset^t$ ($P_S^t$) denote the  Bayesian posterior probability that the state is 1 conditional upon no signal (some signal) being received at time $t$.

\begin{lemma}\label{relativeLearningUtilityLemma}
	If $p_1>p_0$, then $P_1^t\geq P_0^t$  and  $P_1^t\geq P_\emptyset^t$.
\end{lemma}

\noindent{\bf Proof of Lemma \ref{relativeLearningUtilityLemma}}

Let $s^t$ denote the state of the signal at period $t$. That $P_1^t\geq P_0^t$ holds when $t=1$ is easy to check from Bayes rule, given that $p_1> p_0$ and
$\mu\leq 1/2$.
Now suppose $P_1^t\geq P_0^t$ for some $t$. Then by the law of total probability, it follows that
\begin{align*}
P_1^{t+1} &= Pr(s^t = 0 | s^{t+1} = 1) P_0^t + Pr(s^t=1 | s^{t+1} = 1) P_1^t \\
&= \frac{p_0\mu Pr(s^t = 0)}{p_0\mu Pr(s^t = 0) + p_1 (1- \mu) Pr(s^t=1)} P_0^t + \frac{p_1 (1- \mu) Pr(s^t=1)}{p_0\mu Pr(s^t = 0) + p_1 (1- \mu) Pr(s^t=1)} P_1^t
\end{align*}
Similarly,
\[P_0^{t+1} = \frac{p_0(1-\mu) Pr(s^t = 0)}{p_0(1-\mu) Pr(s^t = 0) + p_1 \mu Pr(s^t=1)} P_0^t + \frac{p_1 \mu Pr(s^t=1)}{p_0(1-\mu) Pr(s^t = 0) + p_1 \mu Pr(s^t=1)} P_1^t\]
Since $P_1^t \geq P_0^t$ by the inductive hypothesis, it suffices to show that
\[\frac{p_1 (1- \mu) Pr(s^t=1)}{p_0\mu Pr(s^t = 0) + p_1 (1- \mu) Pr(s^t=1)} \geq \frac{p_1 \mu Pr(s^t=1)}{p_0(1-\mu) Pr(s^t = 0) + p_1 \mu Pr(s^t=1)}\]
i.e., that
\[\frac{1}{1 + \frac{p_0}{p_1}\frac{Pr(s^t = 0)}{Pr(s^t=1)}\frac{\mu}{1-\mu}} \geq \frac{1}{1 + \frac{p_0}{p_1}\frac{Pr(s^t = 0)}{Pr(s^t=1)}\frac{1-\mu}{\mu}}\]
which follows, since $\mu \leq 1-\mu$.

\ \ \ \ To see that $P_1^t \geq P_{\emptyset}^t$, note that it suffices to prove that $P_S^t \geq P_{\emptyset}^t$, since $P_S^t$ is a convex combination of $P_1^t$ and $P_0^t$, and we just proved $P_1^t \geq P_0^t$. Now the statement follows directly from part 1 of Proposition \ref{facts}.\eproof

\medskip
\noindent {\bf Proof of Proposition \ref{relativeLearning}:}

First, note that we can focus on the case in which $p_1\neq p_0$ as otherwise there is nothing to be learned from signal survival, and agent $C$ does as well as $B$.
Without loss of generality we take $p_1>p_0$.   Similarly, if $\mu=1/2$, then all learning is from survival and $S$ does as well as $B$, and so we can take $\mu<1/2$.

Note that by Lemma \ref{relativeLearningUtilityLemma},  $P^t_1\geq P^t_0$  and  $P^t_1\geq P^t_\emptyset$. In order for $B$ to do strictly better in expectation than the other agents, it must be that $P^t_1>1/2$ and at least one of $P^t_0$ and $P^t_\emptyset$ are less than 1/2.
To see this note that if all three are on the same side of 1/2, then they must lie on the same side as the prior.\footnote{They cover three disjoint events whose union is all possibilities, and so the overall probability of a 1 is a convex combination of these conditionals, and so it is impossible to have them all weakly and some strictly greater (or all weakly and some strictly less) than the prior.}   If $\theta\neq 1/2$ then $N$ gets the same payoff as $B$.
If $\theta=1/2$,  then for all three to lie on the same side of the prior it must be that $p_1=p_0$,
in which case there is nothing learned from survival and $C$ does as well as $B$ in expectation.

Thus, $P^t_1>1/2$ and at least one of $P^t_0$ and $P^t_\emptyset$ are less than 1/2.
If it is just $P^t_\emptyset$ that is less than 1/2, then $S$ guesses the same as $B$ (or equivalently in expected payoff terms).
Thus, we need $P^t_0<1/2$ to have a difference.

If is just $P^t_0$ that is less than 1/2, then $C$ guesses the same as $B$ except if $\theta \leq 1/2$.  But for such a $\theta$, it must be that $P^t_\emptyset \leq 1/2$ and so $C$ guesses as well as $B$.

So, consider the case in which $P^t_1>1/2$ and $P^t_0<1/2$ and $P^t_\emptyset <  1/2$. For $C$ to guess differently than $B$, it must be that $\theta\geq 1/2$.

We can compute the expected payoff's for the three most relevant agents for this remaining case (we ignore $N$ now, since in these conditions it is dominated by one of the others) for a given $(p_1, p_0, \mu,\theta)$ satisfying the above constraints.

Letting $U_B, U_C, U_S$ be the expected payoffs of agents $B, C$\footnote{$C$ has expected payoff
	$U_C = Pr(s_t=1)P^t_1 + Pr(s_t=0)(1-P^t_0)
	+(1-Pr(s_t=1)-Pr(s_t=0)) \left(I_{\theta > 1/2} P^t_\emptyset
	+ I_{\theta = 1/2}  1/2 \right).$
	The expression in the main text is obtained by noting that the worst ratio for this compared to $B$ will be in cases for which $\theta>1/2$} and $S$ respectively, it follows that
\begin{align*}
U_B &= Pr(s_t=1)P^t_1 + Pr(s_t=0)(1-P^t_0)  + (1-Pr(s_t=1)-Pr(s_t=0)) (1-P^t_\emptyset) \\
U_C &= Pr(s_t=1)P^t_1 + Pr(s_t=0)(1-P^t_0) + (1-Pr(s_t=1)-Pr(s_t=0)) P^t_\emptyset \\
U_S &= Pr(s_t=1)P^t_1 +  Pr(s_t=0)P^t_0 + (1-Pr(s_t=1)-Pr(s_t=0))(1-P^t_\emptyset)
\end{align*}
First, note that if $p_0<1$ and $\mu>0$, then as $t\rightarrow \infty$, then  $Pr(s_t=\emptyset) \rightarrow 1$ and $P^t_\emptyset \rightarrow \theta$, in which case the ratio of $B$'s payoff to $C$'s goes to 1. If $p_0<1$ and $\mu=0$, then $B$ does as well as $S$ for every $t$. If $p_1=p_0=1$, then $B$ does as well as $C$ for every $t$. These facts together establish the last claim in the proposition that as $t \to \infty$, the ratio $\frac{U_B}{\max\{U_S, U_C\}} \to 1$.

That the ratio is bounded above by 3/2 can be seen as follows.
Since $\theta\geq 1/2$ and $p_1> p_0$, it follows that
\[
Pr(s_t=1) \geq  Pr(s_t=0),   \ \ \   P^t_1 \geq  (1-P^t_0),  {\rm  \ \  and \  so \ \ }  Pr(s_t=1)P^t_1 \geq  Pr(s_t=0)(1-P^t_0).
\]
Then if $Pr(s_t=0)(1-P^t_0)\leq   (1-Pr(s_t=1)-Pr(s_t=0))(1-P^t_\emptyset)$ it follows that $U_S\geq  U_B 2/3$.
If $Pr(s_t=0)(1-P^t_0)\geq   (1-Pr(s_t=1)-Pr(s_t=0))(1-P^t_\emptyset)$ then it follows that $U_C\geq  U_B 2/3$.

To complete the proof, we compute
\[\max_{p_1, p_0, \theta, \mu \in [0, 1]} \frac{U_B}{\max\{U_S, U_C\}}.\]
for $t=1$.   We can rewrite the payoffs of agents $B$, $S$ and $C$ in the case $P_1^1 > 1/2$ and $P_0^1< 1/2$ and $P_{\emptyset}^1 < 1/2$ as follows:
\begin{align*}
U_B &= \theta p_1 (1 - \mu) + (1 - \theta) (1 - p_0 \mu) \\
U_C &= \theta (1 - p_1 \mu) + (1- \theta) p_0 (1 - \mu) \\
U_S &= \theta p_1                + (1 - \theta) (1 - p_0)
\end{align*}
where
\begin{align}
\theta p_1 \mu &\leq p_0 (1 - \theta) (1 - \mu) \label{bbeatss} \\
\theta (1 - p_1) &\leq (1 - \theta) (1 - p_0) \label{bbeatsc}\\
\theta &\geq 1/2 \label{thetarange}\\
\mu &\leq 1/2 \label{murange}\\
p_1 &\geq p_0 \label{survivalcond}\\
p_1, p_0, \mu, \theta &\in [0, 1] \label{paramrange}.
\end{align}

\textbf{Case 1}:  $U_S \leq U_C$.

This condition can be rewritten as
\begin{equation}\label{cbeatss}\theta(p_1\mu + (p_1-1)) \leq (1 - \theta)(p_0(1 - \mu) + (p_0 - 1))\end{equation}
The program with this additional constraint can be written as
\begin{align*}
\max_{p_1, p_0, \theta, \mu \text{ satisfy \ref{bbeatss}-\ref{cbeatss}}} \frac{U_B}{U_C} &\equiv \max_{p_1, p_0, \theta, \mu \text{ satisfy \ref{bbeatss}-\ref{cbeatss}}} \frac{\theta p_1 (1 - \mu) + (1 - \theta) (1 - p_0 \mu)}{\theta (1 - p_1 \mu) + (1- \theta) p_0 (1 - \mu)}\\
&= \max_{p_1, p_0, \theta, \mu \text{ satisfy \ref{bbeatss}-\ref{cbeatss}}} \frac{\theta p_1 + (1 - \theta) - \mu(\theta p_1 + (1 - \theta) p_0)}{\theta + (1 - \theta) p_0 - \mu (\theta p_1 + (1 - \theta) p_0)}\\
&\leq \max_{p_1, p_0, \theta, \mu \text{ satisfy \ref{bbeatss}-\ref{cbeatss}}} \frac{\theta + (1 - \theta) 2p_0 - 2\mu (\theta p_1 + (1 - \theta) p_0)}{\theta + (1 - \theta) p_0 - \mu (\theta p_1 + (1 - \theta) p_0)}
\end{align*}
where the inequality is from rearranging constraint \ref{cbeatss}, as $\theta p_1 + (1 - \theta) \leq \theta + (1-\theta)2p_0 - \mu(\theta p_1 + (1- \theta) p_0)$, and plugging this into the numerator. It is easily verified that the above ratio is decreasing in $\mu$ for any values of the remaining parameters \footnote{$\frac{d}{dx} \frac{A - 2x}{B - x} \leq 0$ if $A \leq 2 B$ and $A, B > 0$.}. Moreover, reducing $\mu$ to 0 only relaxes constraints \ref{bbeatss}, \ref{murange} and \ref{cbeatss}, and leaves the other constraints unaffected. Therefore,
$$
\max_{p_1, p_0, \theta, \mu \text{ satisfy \ref{bbeatss}-\ref{cbeatss}}} \frac{U_B}{U_C} \leq  \max_{p_1, p_0, \theta \text{ satisfy \ref{bbeatsc}-\ref{cbeatss}}} \frac{\theta + (1 - \theta) 2p_0}{\theta + (1 - \theta) p_0}
$$

It is clear that smaller values of $\theta$ increase this ratio, and by constraint \ref{thetarange}, the smallest value of $\theta$ is $\frac{1}{2}$. But while reducing $\theta$ down to $\frac{1}{2}$ for given $p_1$ and $p_0$ relaxes constraint \ref{bbeatsc}, doing so may violate constraint \ref{cbeatss}. We therefore separately consider the cases where either \ref{cbeatss} or \ref{thetarange} bind, since at least one of them must at the optimum.

\textbf{ Subcase 1:} \ref{cbeatss} is satisfied with equality, i.e., $\theta(1 - p_1) = (1 - \theta) (1 - 2p_0)$. Plugging this in, the objective then becomes $2\frac{1 + \theta(p_1-1)}{1 + \theta(p_1-1) + \theta}$, which is decreasing in $\theta$, so it is optimal to set $\theta=\frac{1}{2}$. The objective is then $2 \frac{1+p_1}{2+p_1} \leq 4/3$.
Note that at $p_1=1, p_0=\frac{1}{2}, \theta=\frac{1}{2}, \mu=0$, $\frac{U_B}{U_C} = \frac{4}{3}$, so this upper bound is tight.

\textbf{ Subcase 2:} $\theta=1/2$. Then $$\frac{U_B}{U_C} = \frac{p_1 + 1 - \mu(p_1+p_0)}{p_0+1 - \mu(p_1+p_0)},$$
which is weakly increasing in $\mu$ by constraint \ref{survivalcond}. Constraint \ref{cbeatss} can be rearranged to be $$\mu \leq \frac{2p_0 - p_1}{p_1 + p_0},$$
which, first, implies that $$\frac{U_B}{U_C} \leq \frac{2(p_1-p_0) + 1}{(p_1-p_0)+1},$$ and second, along with the condition that $\mu \geq 0$, implies that $$p_0 \geq \frac{p_1}{2}.$$
Since $\frac{2(p_1-p_0) + 1}{(p_1-p_0)+1}$ is decreasing in $p_0$, this expression is maximized under the given constraints when $p_0 = \frac{p_1}{2}$. Therefore, $\frac{U_B}{U_C} \leq \frac{p_1 + 1}{\frac{p_1}{2}+1}$, which is maximized when $p_1=1$ and equals $4/3$.

\textbf{Case 2:}  $U_S \geq U_C$.
The new constraint is
\begin{equation}\label{sbeatsc}\theta(p_1\mu + (p_1-1)) \geq (1 - \theta)(p_0(1 - \mu) + (p_0 - 1)),\end{equation}
and the relevant maximization program is

$$\max_{p_1, p_0, \theta, \mu \text{ satisfy \ref{bbeatsc}-\ref{paramrange}, \ref{sbeatsc}}} \frac{U_B}{U_S} \equiv \max_{p_1, p_0, \theta, \mu \text{ satisfy \ref{bbeatsc}-\ref{paramrange}, \ref{sbeatsc}}} \frac{\theta p_1 (1 - \mu) + (1 - \theta) (1 - p_0 \mu)}{\theta p_1 + (1- \theta) (1 - p_0)}.$$

Notice that the ratio $\frac{\theta p_1 (1 - \mu) + (1 - \theta) (1 - p_0 \mu)}{\theta p_1 + (1- \theta) (1 - p_0)}$ is linear and decreasing in $\mu$, and the constraints are linear in $\mu$ as well. Constraint \ref{bbeatss} only places an upper bound on $\mu$, so it is not relevant in pinning down this value at the optimum. On the other hand, constraint \ref{sbeatsc}, which can be rewritten as
\[2p_0(1-\theta) - p_1\theta - (1-2\theta) \leq (p_1\theta + p_0(1-\theta))\mu\]
and the constraint that $\mu \geq 0$ are relevant. There are two cases:

\textbf{Subcase 1: $2p_0(1-\theta) - p_1\theta - (1-2\theta) \geq 0$, $\mu=\frac{2p_0(1-\theta) - p_1\theta - (1-2\theta)}{p_1\theta + p_0(1-\theta)}$}.

Then
\begin{align*}
\frac{U_B}{U_S} &= \frac{\theta p_1 + (1 - \theta) - \mu (\theta p_1 + (1 - \theta) p_0)}{\theta p_1 + (1 - \theta) - (1 - \theta) p_0} \\
&=  \frac{\theta p_1 + (1 - \theta) - 2(1-\theta) p_0+ p_1\theta + (1 - 2\theta)}{\theta p_1 + (1 - \theta) - (1 - \theta) p_0} \\
&= \frac{2 \theta p_1 + (2 - 3\theta) - 2(1-\theta)p_0 }{\theta p_1 + (1 - \theta) - (1 - \theta) p_0} \\
&= 2 \frac{ \theta p_1 + (1 - \frac{3}{2}\theta) - (1-\theta)p_0 }{\theta p_1 + (1 - \theta) - (1 - \theta) p_0}
\end{align*}
Clearly, the ratio is decreasing in $\theta$, and moreover, decreasing $\theta$ only relaxes constraints \ref{bbeatss} and \ref{bbeatsc}. Therefore, constraint \ref{thetarange} binds and $\theta = \frac{1}{2}$ at the optimum, so

\[\frac{U_B}{U_S} = 2 \frac{p_1 + \frac{1}{2} - p_0}{p_1 + 1 - p_0}\]

Since $\mu = \frac{2p_0 - p_1}{p_1 + p_0}$ at $\theta=\frac{1}{2}$, constraints \ref{bbeatss} and \ref{bbeatsc} reduce to just $p_1 \geq p_0$. Since the ratio is increasing in $p_1 - p_0$, the only binding constraint is that $\mu \geq 0$, i.e., $2p_0 \geq p_1$. Therefore at the optimum, $p_1=1$, $p_0 = \frac{1}{2}$, $\mu=0$, $\theta=\frac{1}{2}$, and $\frac{U_B}{U_S} = \frac{4}{3}$.

\textbf{Subcase 2: $2p_0(1-\theta) - p_1\theta - (1-2\theta) \leq 0$, $\mu=0$.} In this case, the problem reduces to
\[\max_{p_1, p_0, \theta, \text{ satisfy \ref{bbeatsc},\ref{thetarange},\ref{survivalcond}, \ref{paramrange}, \ref{sbeatsc}}} \frac{\theta p_1  + (1 - \theta)}{\theta p_1 + (1- \theta) (1 - p_0)},\]
which is decreasing in $\theta$. Now
\begin{align*}
2p_0(1-\theta) - p_1\theta - (1-2\theta) &\leq 0 \\
\iff \theta(2 - 2p_0 - p_1) &\leq 1 - 2p_0
\end{align*}
Suppose $2 - 2p_0 - p_1 < 0$. Since $1 - 2p_0 \leq 1 - 2p_0 +(1-p_1) = 2 - 2p_0 - p_1 < 0$, it follows that $\theta (2 - 2p_0 - p_1) \geq \theta (1 - 2p_0) \geq 1 - 2p_0$.
Therefore, the only way to satisfy the constraint is if $p_1=1$ and $\theta=1$, in which case $\frac{U_B}{U_S} = 1$.

If $2 - 2p_0 - p_1 \geq 0$, then $\theta = \frac{1}{2}$
at the optimum, and so the constraint in this sub-subcase becomes $2p_0 \leq p_1$, while the objective function
is $\frac{p_1+1}{p_1 + 1 - p_0}$.
This constraint binds at the optimum and again the optimal value is $\frac{4}{3}$ at $p_1=1$
and $p_0=\frac{1}{2}$.\eproof

\medskip

The next lemma is useful in the proof of the second part of Proposition \ref{multiSignalRelativeLearning}.

\begin{lemma}
	\label{Bayes}
	Let $p_0=p_1$ and suppose a Bayesian agent with prior $\theta$ on the state being $1$ observes $m$ signals of which $k$ are $1$s.
	Then the agent's posterior that the state is 1 is
	\[
	\frac{\theta X_1(t)^{k}(1 - X_1(t))^{m-k}}{\theta X_1(t)^{k}(1 - X_1(t))^{m-k} + (1-\theta) (1-X_0(t))^{k}X_0(t)^{m-k}}. %= \frac{\theta(1 + \delta^t)^k}{\theta(1 + \delta^t)^k + (1-\theta)(1 - \delta^t)^k}.
	\]
\end{lemma}
The proof is direct and omitted.

\medskip

\noindent{\bf Proof of Proposition \ref{multiSignalRelativeLearning} :}

We proceed by cases for different values of the parameters.  We concentrate on situations in which $\mu<1/2$ since if $\mu=1/2$ then content is completely uninformative and  the result is direct.

\textbf{Case 1}: $\mu=0$. Suppose without loss of generality that $p_0 \leq p_1$. Any signal that reaches the agent is perfectly informative of the state, so a threshold for learning for agent B is the threshold for at least one signal to survive, which  (following the logic of the proofs above) is $\frac{1}{p_1^t}$. This is the threshold for learning for S when $p_0<p_1$ (survival counts separate the states) and for C when $p_0=p_1$ (every surviving message reveals the state).

\textbf{Case 2}: $p_1=p_0$ and $\mu>0$. By Proposition \ref{multipath1}, the threshold for learning for agent B is $\frac{1}{p_1^t (1 - 2\mu)^{2t}}$.  In this case there is no information from signal survival, and by Lemma \ref{Bayes}, agent B's posterior is the same as agent C's  posterior. Therefore, agent $C$ has the same threshold for learning as B.

\textbf{Case 3}:  $p_0 \neq p_1$ and $\mu>0$.
Without loss of generality let $p_1>p_0$. Then $\tau(t) = \frac{1}{P^t_{1S}}$ is a threshold for learning for an agent conditioning only on signal survival, as shown in the proof of Lemma \ref{learn_survival}.
Let $b(t)$ denote the beliefs of agent B after observing the outcome of $n(t)$ original sources of information sent along chains of depth $t$.
Since agent $B$ conditions on survival and signal content, $plim \ b(t) \to 1$ or 0 whenever $n(t)/\tau(t) \to \infty$.
When  $n(t)/\tau(t) \to 0$, then the probability of even a single signal surviving to reach the
agent approaches 0. This holds regardless of the starting state by Lemma \ref{facts} part 2, so   $plim \  b(t) \to \theta$. Therefore, agent B and S have the same thresholds for learning in this case.\footnote{Strictly speaking, we only showed that they share a common threshold, but it is easy to see that being a threshold for learning for B, for S or for neither partitions the space of functions on $\mathbb{N} \to \mathbb{N}$.}.\eproof

\end{document}